%% file: ArXiv_Replace.tex
\documentclass{article}

\usepackage{arxiv}

\usepackage[utf8]{inputenc}
\usepackage[T1]{fontenc}
\usepackage{times}
\usepackage{microtype}
\usepackage{amsmath}
\usepackage{amssymb}
\usepackage{amsfonts}
\usepackage{amsthm}
\usepackage{mathrsfs}
\usepackage{bm}
\usepackage{nicefrac}
\usepackage{graphicx}
\usepackage{booktabs}
\usepackage{multirow}
\usepackage{arydshln}
\usepackage{lscape}
\usepackage{float}
\usepackage{here}
\usepackage{wrapfig}
\usepackage{caption}
\usepackage{afterpage}
\usepackage{algorithm}
\usepackage{algorithmic}
\usepackage{enumitem}
\usepackage{comment}
\usepackage{color}
\usepackage{xcolor}
\usepackage{soul}
\usepackage{tikz}
\usepackage{ascmac}
\usepackage{bbding}
\usepackage{stackengine}
\usepackage{perpage}
\MakeSorted{figure}
\MakeSorted{table}
\usepackage{url}
\usepackage[square]{natbib}
\usepackage{chapterbib}
\usepackage{thmtools}
\usepackage[colorlinks=true,allcolors=blue]{hyperref}
\usepackage[capitalize,noabbrev]{cleveref}

\allowdisplaybreaks[4]

\tikzset{%
mynode/.style={circle,minimum width=.5ex, fill=none,draw},
myfillnode/.style={circle,minimum width=.5ex, fill=lightgray,draw},
}

\theoremstyle{plain}
\newtheorem{theorem}{Theorem}
\newtheorem{lemma}[theorem]{Lemma}

\newtheorem{corollary}[theorem]{Corollary}

\newtheorem{assumption}{Assumption}
\theoremstyle{definition}
\newtheorem{definition}[theorem]{Definition}

\def\defeq{\mathrel{\ensurestackMath{\stackon[1pt]{=}{\scriptscriptstyle\Delta}}}}
\newcommand{\indep}{\perp \!\!\! \perp}

\newcommand{\jin}[1]{}

\newcommand{\yuta}[1]{#1}

\providecommand{\keywords}[1]{\par\noindent\textbf{Keywords:} #1}

\title{Beyond Conditional Average Effects: Quantifying Positive and Negative Causal Contributions}
\renewcommand{\shorttitle}{Quantifying Positive and Negative Causal Contributions}
\date{\today}
\author{
Yuta Kawakami\textsuperscript{1} \quad and \quad Jin Tian\textsuperscript{2}\\[2ex]
\textsuperscript{1}Mohamed bin Zayed University of Artificial Intelligence; E-mail: \texttt{Yuta.Kawakami@mbzuai.ac.ae}\\
\textsuperscript{2}Mohamed bin Zayed University of Artificial Intelligence; E-mail: \texttt{Jin.Tian@mbzuai.ac.ae}
}

\begin{document}

\maketitle

\begin{abstract}
Understanding causal effect heterogeneity is important for evaluating interventions and informing individual decisions. The conditional average causal effect (CACE) summarizes how average effects vary across covariate-defined subpopulations, but it can obscure positive and negative individual effects within these subpopulations. Existing work on effect heterogeneity beyond CACE has primarily focused on binary treatments and outcomes. In this paper, we introduce two heterogeneity measures for binary treatments and continuous outcomes: P-CACE and N-CACE, that quantify the positive and absolute negative contributions of individual causal effects to CACE, respectively. Their difference equals CACE, revealing contributions that may otherwise be obscured by averaging. We further introduce their counterparts for continuous treatments and outcomes, P-CPICE and N-CPICE. They leverage stochastic interventions to broaden the range of causal questions that researchers can investigate. We establish identification and bounding theorems for these measures. An application to a medical dataset illustrates how these measures provide information about individual causal effects beyond conditional averages.
\end{abstract}

\keywords{Causal effect heterogeneity \and CACE \and Positive and negative causal contributions}

\section{Introduction}
\label{sec1}

\input{1_Introduction}

\section{Notation and Background}
\label{sec-2}
\input{2_Background}

\section{Positive and Negative Contributions to CACE}
\input{3_PN-CACE}

\section{Positive and Negative Contributions to CPICE}
\input{4_PN-CPICE}

\section{Application to a Real-World Dataset}
\label{sec7}
\input{7_Application}

\section{Conclusion}

\input{8_Conclusion}

\bibliographystyle{plainnat}
\bibliography{icml2025}

\clearpage
\appendix
\onecolumn

\section*{\LARGE Appendix to ``Beyond Conditional Average Effects: Quantifying Positive and Negative Causal Contributions''}
\vspace{0.5cm}

\section{Summary Table of Causal Effect Heterogeneity Measures}
\label{appT}

\input{A0_Table}

\section{An Example of P-CACE and N-CACE}
\label{app-ex}

\input{App_Example}

\section{Estimation Methods}
\label{appA}
\input{A2_LL}

\section{Numerical Experiments}
\label{sec6}
\input{6_Experiment}

\section{Proofs}
\label{app_proof}
\input{A3_Proof}

\end{document}

%% file: 1_Introduction.tex

There is 
a growing interest in understanding how the causal effect of a treatment ($X$) on an outcome ($Y$) varies across different individuals or subgroups within a population \citep{Xie2013}. 
To investigate this heterogeneity, researchers often employ the \emph{conditional average causal effect (CACE)}, also known as conditional average treatment effect (CATE), given by $\mathbb{E}[Y_1-Y_0|W=w]$ 
where $Y_x$ denotes the potential outcome under treatment $X=x$ and $W$ are the subjects' covariates (e.g. gender, age, and race). 
CACE allows for the assessment of causal effects within specific subpopulations defined by observed covariates of the subjects \citep{Athey2016,Shalit2017,Athey2019,Kunzel2019,Wager2018,Jacob2021,Zhang2022,Singh2023,Kawakami2024b}.
Many studies on this heterogeneity assume the additive noise model $Y:=f_Y'(X,W)+U_Y$, such as the partially linear regression model $Y:=f_Y^1(W)X+f_Y^2(W)+U_Y$  \citep{Robinson1988,Chernozhukov2018}.
In this model, the whole treatment effect heterogeneity is fully explained away by the observed covariates $W$ using CACE. 
However, 
in many cases, considerable  heterogeneity remains unexplained
by CACE, likely due to the interaction between $X$ and unobserved variables $U_Y$  \citep{Gadbury2004,Poulson2012}.
For example, consider the relationship $Y:=X+W+XU_Y$, $U_Y \sim {\cal N}(0,1)$. {Given $W=0$, we have that the individual causal effect $Y_{1}-Y_{0}=1+U_Y$   still depends on $U_Y$  and could be positive or negative, 
even though CACE $\mathbb{E}[Y_1-Y_0|W=0]=1$.}

Research on individual causal effect heterogeneity beyond CACE primarily focuses on binary treatments and outcomes  
\citep{Gadbury2000,Gadbury2001,Albert2005,Poulson2012,Brand2013,Zhang2013,Yin2018,Post2023,Wu2024}.
Under the framework of principal stratification for binary treatments and outcomes  \citep{Frangakis2002}, individuals can be categorized into four distinct groups \citep{Ben-Michael2024}. To quantify causal effect heterogeneity across these groups, several measures have been proposed. Notably,   
\citet{Shen2013} introduced the (conditional) \emph{treatment harm rate} (THR) as $\mathbb{P}(Y_0=1,Y_1=0|W=w)$ 
and the (conditional) \emph{treatment benefit rate} (TBR) as $\mathbb{P}(Y_0=0,Y_1=1|W=w)$. 
THR is also called the fraction negatively affected (FNA) in machine learning \citep{Kallus2022,Li2023}.
Extensions of THR and TBR to continuous outcomes have also been proposed \citep{Shen2013,Zhang2013,Yin2018,Kallus2022}.




{In this paper, we focus on continuous outcomes and quantify positive and negative causal contributions that average effects can obscure. Individuals with the same observed covariates may experience effects of opposite signs, so a small conditional average effect can reflect substantial positive and negative contributions that largely cancel. Separating these contributions therefore provides information about individual causal effects that CACE alone does not convey.}

First, we define a new classification of individuals into four groups under binary treatment and continuous outcome, extending the classification by \citet{Ben-Michael2024}.
We then introduce two measures, P-CACE and N-CACE, that quantify the positive and absolute negative contributions of individual causal effects to CACE, respectively.
These measures yield a natural decomposition: CACE equals P-CACE minus N-CACE.

Next, we extend this analysis to 
continuous treatments and  outcomes employing stochastic interventions, which allow us to investigate more flexible scenarios of continuous treatments \citep{Munoz2012}. Within this framework, we classify individuals into four  groups and introduce  P-CPICE and N-CPICE. 
These measures quantify the positive and absolute negative contributions to the causal contrast between two stochastic interventions, yielding an analogous decomposition of the corresponding average effect.

We establish identification and bounding theorems for each measure we introduce. 
Finally, we 
illustrate their application to a real-world dataset from medicine.
We provide a summary table comparing the applicable settings of existing causal effect heterogeneity measures to ours in Appendix~\ref{appT}.

%% file: 2_Background.tex
We represent each variable with a capital letter $(X)$ and its realized value with a small letter $(x)$.
Let $\mathbb{I}(x)$ be an indicator function that takes $1$ if $x$ is true and $0$ if $x$ is false.
Denote $\Omega_Y$ the domain of $Y$,
$\mathbb{E}[Y]$  the expectation of $Y$, and 
$\mathbb{P}(Y< y)$ the cumulative distribution function (CDF) of the continuous variable $Y$.
We write $g(x)={\cal O}(h(x))$ as $x \rightarrow \infty$ if there exists a positive real number $M$ and a real number $D$ such that $|g(x)| \leq Mh(x)$ for all $x \geq D$.
$X_N=O_p(a_N)$ means the set of values $X_N/a_N$ are stochastically bounded and 
$X_N=o_p(a_N)$ means the 
values $X_N/a_N$ converge to zero in probability as $N \rightarrow \infty$.

\paragraph{Structural causal models.} 
We use the language of Structural Causal Models (SCM) as our basic semantic and inferential framework \citep{Pearl09}.
An SCM ${\cal M}$ is a tuple $\left<{\boldsymbol U},{\boldsymbol V}, {\cal F}, \mathbb{P}_{\boldsymbol U} \right>$, where ${\boldsymbol U}$ is a set of exogenous (unobserved) variables following a joint distribution $\mathbb{P}_{\boldsymbol U}$, and ${\boldsymbol V}$ is a set of endogenous (observable) variables whose values are determined by structural functions ${\cal F}=\{f_{V_i}\}_{V_i \in {\boldsymbol V}}$ such that $v_i:= f_{V_i}({\mathbf{pa}}_{V_i},{\boldsymbol u}_{V_i})$ where ${\mathbf{PA}}_{V_i} \subseteq {\boldsymbol V}$ and $U_{V_i} \subseteq {\boldsymbol U}$. 
Each SCM ${\cal M}$ induces an observational distribution $\mathbb{P}_{\boldsymbol V}$ over ${\boldsymbol V}$.  
An \emph{atomic intervention} of setting a set of endogenous variables ${\boldsymbol X}$ to constants ${\boldsymbol x}$, denoted by $do({\boldsymbol x})$, replaces the original equations of ${\boldsymbol X}$ by ${\boldsymbol X} :={\boldsymbol x}$ and induces a \textit{sub-model}  ${\cal M}_{\boldsymbol x}$.
We denote the \emph{potential outcome} $Y$ under intervention $do({x})$ by $Y_{{x}}({\boldsymbol u})$, which is the solution of $Y$ with ${\boldsymbol U}={\boldsymbol u}$ in the sub-model ${\cal M}_{x}$. 
A \emph{stochastic intervention} $do({\boldsymbol X}^{\pi})$ replaces the original equations of ${\boldsymbol X}$ by setting ${\boldsymbol X}$ to random variables ${\boldsymbol X}^{\pi}$ with a probability distribution function (PDF) $\pi({\boldsymbol x})$, inducing a \textit{sub-model} ${\cal M}_{{\boldsymbol X}^{\pi}}$. 
The potential outcome $Y_{{\boldsymbol X}^{\pi}}({\boldsymbol u})$ denotes a solution of $Y$ with ${\boldsymbol U}={\boldsymbol u}$ in the sub-model ${\cal M}_{{\boldsymbol X}^{\pi}}$.




When studying the causal effect heterogeneity, researchers often consider the following SCM ${\cal M}$:
\begin{align}
Y:=f_Y(X,W,U_Y), X:=f_X(W,U_X), W:=f_W(U_W),
\end{align}
where {$Y$ is an 
outcome variable, $X$ is a 
treatment variable, $W$ are a set of covariates, and } 
$U_Y$, $U_X$, and $U_W$ are latent exogenous variables following a joint distribution $\mathbb{P}(U_Y, U_X, U_W)$.

\paragraph{Causal effects measures.}
We explain popular causal effect measures 
used in the previous work.

\paragraph{Binary treatment.} 
For a binary treatment, the \emph{individual causal effect (ICE)} is defined as {$\text{\normalfont ICE}(\boldsymbol{u})\defeq Y_1(\boldsymbol{u})-Y_0(\boldsymbol{u})$.}
The \emph{average causal effect (ACE)} is defined as $\text{\normalfont ACE}\defeq\mathbb{E}[Y_1-Y_0]$.
For any $w \in \Omega_W$, the \emph{conditional average causal effect (CACE)} is defined as $\text{\normalfont CACE}(w)\defeq\mathbb{E}[Y_1-Y_0|W=w]$, which measures causal effect heterogeneity across subpopulations specified by the  subjects' observed covariates. 



Substantial heterogeneity may persist even after accounting for the covariates. To investigate individual causal effect heterogeneity, researchers divide subjects into four groups by the joint values of the potential outcomes for binary treatment and outcome, known as \emph{principal stratification} \citep{Frangakis2002}. 
This categorization is widely used in the causal inference area and has various names, e.g., compliers, always-takers, defiers, and never-takers in non-compliance problems \citep{Imbens1994}.
The following categorization is used in \cite{Ben-Michael2024,Li2023}:
\begin{definition} \label{def-ps}
We call subjects:
\begin{enumerate}
\setlength{\parskip}{0cm}
\setlength{\itemsep}{0.2cm}
\item for which $\{Y_0=0,Y_1=1\}$, {\it useful treatment group},
\item for which $\{Y_0=1,Y_1=0\}$, {\it harmful treatment group},
\item for which $\{Y_0=0,Y_1=0\}$, {\it useless treatment group}, 
\item for which $\{Y_0=1,Y_1=1\}$, {\it harmless treatment group}.
\vspace{-0cm}
\end{enumerate}
\end{definition}
\noindent Based on the principal stratification, \citet{Shen2013} introduced the \emph{treatment harm rate (THR)}, $\text{\normalfont THR}(w)\defeq\mathbb{P}(Y_0=1,Y_1=0|W=w)$, 
and the \emph{treatment benefit rate (TBR)}, $\text{\normalfont TBR}(w)\defeq\mathbb{P}(Y_0=0,Y_1=1|W=w)$, to measure causal effect heterogeneity for binary outcomes.  
THR is also called the \emph{fraction negatively affected (FNA)} in \citep{Kallus2022,Li2023}.

THR and TBR have been extended for use with a continuous outcome.
\citet{Shen2013} and \citet{Zhang2013} use THR and TBR by binarizing potential outcomes $\left(\mathbb{I}(Y_0<y_0),\mathbb{I}(Y_1<y_1)\right)$ given thresholds $y_0$ and $y_1$. \citet{Yin2018} extend THR and TBR by $\text{\normalfont THR}_c(w)\defeq\mathbb{P}(Y_0-Y_1>c|W=w)$ and $\text{\normalfont TBR}_c(w)\defeq\mathbb{P}(Y_1-Y_0>c|W=w)$ given a constant $c$. 
\citet{Kallus2022} defines FNA for 
continuous outcome as $\mathbb{P}(\zeta(W)(Y_1-Y_0)<\delta)$ with a threshold $\delta$ and a weighting function $\zeta$.

{Effect heterogeneity between treated and untreated units is examined in \cite{Pearl2017}, which studies the difference between the average effect of treatment on the treated (ETT)  and the average effect of treatment on the untreated (ETU), given by  $\mathbb{E}[Y_1-Y_0|X=1]$ - $\mathbb{E}[Y_1-Y_0|X=0]$.}

\paragraph{Continuous treatment.}
\citet{Munoz2012} defined the \emph{population intervention causal effect (PICE)} as $\text{\normalfont PICE}\defeq\mathbb{E}[Y_{X^{\pi_1}}-Y_{X^{\pi_0}}]$ 
based on stochastic interventions, which enable one to study more flexible scenarios of continuous treatment. 
For example, \citet{Mauro2020} introduced ``\emph{single shift}'' interventions with random variables $(X^{\pi_0},X^{\pi_1})=(X,X+d)$, {which increase the treatment by a set amount $d$ from the current value,} and ``\emph{double shift}'' $(X^{\pi_0},X^{\pi_1})=(X-d,X+d)$ for some value $d \in \mathbb{R}$. 
See other examples of stochastic interventions in \cite{Mauro2020} and \cite{Diaz2020}.


\paragraph{Probabilities of causation (PoC)}.
THR and TBR are closely related to the framework of probabilities of causation (PoC), especially the probability of necessity and sufficiency (PNS), 
a family of probabilities quantifying whether one event was the real cause of another 
\citep{Pearl1999,Tian2000}.
Recently, \citet{Kawakami2024} defined the conditional PNS for continuous treatment and outcome by $\mathbb{P}(Y_{x_0}< y \leq Y_{x_1}|W=w)$, 
and showed that the conditional PNS is identified under 
the assumption that there are no unobserved confounders (other than $W$) and 
a monotonicity assumption on $Y_x$:
\begin{assumption}[Conditional exogeneity]
\label{ASEXO2}
$Y_x\indep X | W$ for all $x \in \Omega_X$.
\end{assumption}
\begin{assumption}[Conditional monotonicity over $Y_{x}$]\footnote{Assumption~\ref{MONO2} is equivalent to a monotonicity assumption on the function $f_Y$ stated as follows: 
``The function $f_Y(x,w,U_Y)$ is either (i) monotonic increasing on $U_Y$ for all $x \in \Omega_X$ and $w \in \Omega_W$ almost surely w.r.t. $\mathbb{P}_{U_Y}$, or (ii) monotonic decreasing on $U_Y$ for all $x \in \Omega_X$ and $w \in \Omega_W$ almost surely w.r.t. $\mathbb{P}_{U_Y}$" \citep{Kawakami2024}. For example, additive noise model, i.e., $Y:=f(X,W)+U_Y$, and models with multiplicative noise of the form $Y:=f(X,W)\times U_Y$  satisfy this assumption {if $f$ is a positive-valued function}.}
\label{MONO2}
The potential outcome $Y_{x}$ satisfies: for any $x_0,x_1 \in \Omega_X$, $y \in \Omega_Y$, and $w \in \Omega_W$, either $\mathbb{P}(Y_{x_0}< y \leq Y_{x_1}|W=w)=0$ or $\mathbb{P}(Y_{x_1}< y \leq Y_{x_0}|W=w)=0$.
\end{assumption}
\noindent Specifically, under SCM ${\cal M}$ and Assumptions \ref{ASEXO2} and \ref{MONO2},  $\mathbb{P}(Y_{x_0}< y \leq Y_{x_1}|W=w)$ is identifiable by
$\max\{\mathbb{P}(Y<y|X=x_0,W=w)-\mathbb{P}(Y<y|X=x_1,W=w),0\}$ \citep{Kawakami2024}. 


%% file: 3_PN-CACE.tex
In this section, we consider  binary treatment and continuous outcome and 
study two heterogeneity measures, P-CACE and N-CACE, that quantify the positive and absolute negative contributions of individual causal effects to CACE, respectively. We then establish identification and bounding results for these measures.

\subsection{Principal stratification for continuous outcome}
First, we  extend the principal stratification 
in Def.~\ref{def-ps} to binary treatment and continuous outcome.
\begin{definition}[Potential outcome types for a continuous outcome]
\label{def7}
For any $y \in \Omega_Y$, 
we name subjects: 
\begin{enumerate}
\setlength{\parskip}{0cm}
\setlength{\itemsep}{0.2cm}
\item for which $\{Y_0< y,Y_1\geq y\}$, {\it the positively affected subjects at point $y$},
\item for which $\{Y_0\geq y,Y_1< y\}$, {\it the negatively affected subjects at point $y$},
\item for which $\{Y_0\geq y,Y_1\geq y\}$, {\it the positively immutables at point $y$},
\item for which $\{Y_0< y,Y_1< y\}$, {\it the negatively immutables at point $y$}.
\vspace{-0cm}
\end{enumerate}
\end{definition}
\noindent The positively affected subjects at point $y$ are those propelled by the treatment from an outcome less than $y$ to at least $y$ by changing the intervention from $0$ to $1$.
The negatively affected subjects at point $y$ are those propelled by the treatment from an outcome at least $y$ to less than $y$ by changing the intervention from $0$ to $1$.
{Note that a subject is positively affected at some point  if and only if $Y_0<Y_1$, and a subject is negatively affected at some point if and only if $Y_1<Y_0$.}
The immutables at the point $y$ are those whose outcome remains at least $y$ or less than $y$ even after changing the intervention from $0$ to $1$.


The following relationship is useful:
\begin{restatable}{lemma}{Lemmaone}
\label{LEM1}
For any $y \in \Omega_Y$ and $w \in \Omega_W$, we have 
\begin{equation}
\begin{aligned}
&\mathbb{P}(Y_0<y|W=w)-\mathbb{P}(Y_1<y|W=w)\\
&=\mathbb{P}(Y_0<y \leq Y_1|W=w)-\mathbb{P}(Y_1<y\leq Y_0|W=w).
\end{aligned}
\end{equation}
\end{restatable}

\subsection{P-CACE and N-CACE for measuring causal effect heterogeneity}

We introduce the following measures for capturing causal effect heterogeneity unexplained by CACE.
\begin{definition}[P-CACE and N-CACE]
For any $w\in\Omega_W$, we define the positive and absolute negative contributions of individual causal effects to CACE, denoted by P-CACE and N-CACE, respectively, as
\begin{equation}
\begin{aligned}
&\text{\normalfont P-CACE}(w)\defeq\int_{\Omega_Y} \mathbb{P}(Y_0 < y \leq Y_1|W=w)dy,\\ 
&\text{\normalfont N-CACE}(w)\defeq\int_{\Omega_Y} \mathbb{P}(Y_1 < y \leq Y_0|W=w)dy.
\end{aligned}
\end{equation}
The population-level positive and absolute negative contributions are defined as
$\text{\normalfont P-ACE}\defeq\mathbb{E}_W[\text{\normalfont P-CACE}(W)]$ and $\text{\normalfont N-ACE}\defeq\mathbb{E}_W[\text{\normalfont N-CACE}(W)]$.
\end{definition}
\noindent P-CACE is an integration of the probability of the positively affected subjects at point $y$ over $\Omega_Y$.
N-CACE is an integration of the probability of the  negatively affected subjects at point $y$ over $\Omega_Y$.

P-CACE and N-CACE can also be expressed as 
\begin{equation}
\begin{aligned}
\text{\normalfont P-CACE}(w)=\mathbb{E}_{\boldsymbol{U}}\left[\int_{\Omega_Y} \mathbb{I}(Y_0 < y \leq Y_1)dy\Big|W=w\right],
\end{aligned}    
\end{equation}
\begin{equation}
\begin{aligned}
\text{\normalfont N-CACE}(w)=\mathbb{E}_{\boldsymbol{U}}\left[\int_{\Omega_Y} \mathbb{I}(Y_1 < y \leq Y_0)dy\Big|W=w\right].
\end{aligned}    
\end{equation}
We have
\begin{align}\label{eq-aaa}
 \int_{\Omega_Y} \mathbb{I}(Y_0 < y \leq Y_1)dy = \begin{cases}
  \text{\normalfont ICE}(\boldsymbol{u}) & \text{\normalfont ICE}(\boldsymbol{u}) > 0 \\
   0 &\text{\normalfont ICE}(\boldsymbol{u}) \leq 0
 \end{cases},
\end{align}
where $\text{\normalfont ICE}(\boldsymbol{u})= Y_1(\boldsymbol{u})-Y_0(\boldsymbol{u})$. 
Therefore, P-CACE becomes
\begin{equation}
\begin{aligned}\label{eq-125}
&\text{\normalfont P-CACE}(w)=\mathbb{E}[\text{\normalfont ICE}(\boldsymbol{u})|\text{\normalfont ICE}(\boldsymbol{u})>0,W=w]\mathbb{P}(\text{\normalfont ICE}(\boldsymbol{u})>0|W=w).
\end{aligned}    
\end{equation}
Therefore, P-CACE is not the average causal effect among positively affected subjects; it equals that conditional average multiplied by the prevalence of positive ICE values.
Similarly, we have
\begin{equation}
\begin{aligned}\label{eq-1252}
&\text{\normalfont N-CACE}(w)=\mathbb{E}[-\text{\normalfont ICE}(\boldsymbol{u})|\text{\normalfont ICE}(\boldsymbol{u})<0,W=w]\mathbb{P}(\text{\normalfont ICE}(\boldsymbol{u})<0|W=w).
\end{aligned}    
\end{equation}
N-CACE can be interpreted as measuring the (absolute value of) the negative portion of CACE over negatively affected subjects (at some point). Both P-CACE and N-CACE take non-negative values.



\paragraph{Remark.} When the outcome $Y$ is binary, P-CACE and N-CACE reduce to existing measures TBR and THR, respectively. Specifically, $\text{\normalfont P-CACE}(w)= \mathbb{P}(Y_0 < 0 \leq Y_1|W=w)+\mathbb{P}(Y_0 < 1 \leq Y_1|W=w)=\mathbb{P}(Y_0 =0, Y_1=1|W=w)$, and $\text{\normalfont N-CACE}(w)= \mathbb{P}(Y_1 < 0 \leq Y_0|W=w)+\mathbb{P}(Y_1 < 1 \leq Y_0|W=w)=\mathbb{P}(Y_0 =1, Y_1=0|W=w)$.

\subsection{Decomposition of CACE by P-CACE and N-CACE \label{sec-34}}

P-CACE and N-CACE together form CACE. 
We make the following assumption.
\begin{assumption}[Finiteness of P-CACE and N-CACE]
\label{exi1}
We assume $\int_{\Omega_Y} \mathbb{P}(Y_0 < y \leq Y_1|W=w)dy<\infty$ and $\int_{\Omega_Y} \mathbb{P}(Y_1 < y \leq Y_0|W=w)dy<\infty$
for any $w \in \Omega_W$.
\end{assumption}

CACE can then be decomposed into the difference of P-CACE and N-CACE by Lemma~\ref{LEM1} as follows.
\begin{restatable}{proposition}{Propositionone}[Decomposition]
\label{PROP1}
Under SCM ${\cal M}$ and Assumption \ref{exi1},
for any $w \in \Omega_W$, we have
\begin{equation}
\text{\normalfont CACE}(w)=\text{\normalfont P-CACE}(w)-\text{\normalfont N-CACE}(w).
\end{equation}
\end{restatable}

We provide examples to illustrate how P-CACE and N-CACE measure the heterogeneity of causal effects. 
{For simplicity, we first consider models without covariates $W$ and compute P-ACE and N-ACE.}

\paragraph{Example 1 (Homogeneous and positive ICE).}
Consider the model $Y:=X+U_Y$, $U_Y \sim {\cal N}(0,1)$, 
which has homogeneous and positive {ICE~$=1$}.
We have $\text{\normalfont ACE}=\mathbb{E}[Y_1]-\mathbb{E}[Y_0]=1$,
$\text{\normalfont P-ACE}
=\int_{-\infty}^{\infty} \mathbb{P}(U_Y < y \leq U_Y+1)dy=1$, and
$\text{\normalfont N-ACE}
=\int_{-\infty}^{\infty} \mathbb{P}(U_Y+1 < y \leq U_Y)dy=0$.
These results are consistent with the fact that no subjects have negative individual causal effects in this setting.

\paragraph{Example 2 (Heterogeneous ICE).}
We consider the model $Y:=XU_Y$, $U_Y \sim {\cal N}(0,1)$, where half of the subjects have positive ICE and the other half have negative ICE.
Under this setting, 
$\text{\normalfont ACE}=\mathbb{E}[Y_1]-\mathbb{E}[Y_0]=0$,
$\text{\normalfont P-ACE}
=\int_{0}^{\infty} \mathbb{P}(y \leq U_Y)dy>0$, and
$\text{\normalfont N-ACE}
=\int_{0}^{\infty} \mathbb{P}(U_Y < -y)dy>0$. In addition, we have $\text{\normalfont P-ACE}=\text{\normalfont N-ACE}$ 
since $\mathbb{P}(y \leq U_Y)=\mathbb{P}(U_Y < -y)$ holds for $U_Y \sim {\cal N}(0,1)$ for any $y \in [0,\infty]$.  
Even when ACE is zero, P-ACE and N-ACE can reveal offsetting positive and negative contributions of individual causal effects that cancel in the average.

\paragraph{Example 3 (Homogeneous and zero ICE).}
We consider the model $Y:=U_Y$, $U_Y \sim {\cal N}(0,1)$, where all subjects have zero ICE.
Under this setting,
$\text{\normalfont ACE}=\mathbb{E}[Y_1]-\mathbb{E}[Y_0]=0$,
$\text{\normalfont P-ACE}
=\int_{-\infty}^{\infty} \mathbb{P}(U_Y < y \leq U_Y)dy=0$, and
$\text{\normalfont N-ACE}
=\int_{-\infty}^{\infty} \mathbb{P}(U_Y < y \leq U_Y)dy=0$.

ACEs in both Examples 2 and 3 are equal to 0; however, 
they have completely different patterns of effect heterogeneity.
{Our P-ACE and N-ACE are able to measure  the  heterogeneity unexplained by ACE.}
We provide an additional example with covariates in Appendix B.

\subsection{Identification of P-CACE and N-CACE}

P-CACE and N-CACE can be identified by identifying counterfactual probabilities $\mathbb{P}(Y_0 < y \leq Y_1|W=w)$ and $\mathbb{P}(Y_1 < y \leq Y_0|W=w)$.  
The identification of this type of probabilities  is discussed in \citep{Kawakami2024}.
We obtain the following result: 
\begin{restatable}{theorem}{Theoremone}[Identification of P-CACE and N-CACE]
\label{theo1}
Under SCM ${\cal M}$ and Assumptions \ref{ASEXO2}, \ref{MONO2}, and \ref{exi1}, 
P-CACE and N-CACE are identifiable by
\begin{equation}
\begin{aligned}
\label{eq14}
&\text{\normalfont P-CACE}(w)=\int_{\Omega_Y} \max\Big\{\mathbb{P}(Y<y|X=0,W=w)-\mathbb{P}(Y<y|X=1,W=w),0\Big\}dy,
\end{aligned}
\end{equation}
\begin{equation}
\begin{aligned}
\label{eq15}
&\text{\normalfont N-CACE}(w)=\int_{\Omega_Y} \max\Big\{\mathbb{P}(Y<y|X=1,W=w)-\mathbb{P}(Y<y|X=0,W=w),0\Big\}dy.
\end{aligned}
\end{equation}
\end{restatable}
\noindent {An illustration of the computation of P-CACE and N-CACE in terms of the conditional CDF $\mathbb{P}(Y<y|X=0,W=w)$ and $\mathbb{P}(Y<y|X=1,W=w)$ is provided in Figure~\ref{fig:1}.}

\begin{wrapfigure}[16]{r}{0.45\textwidth}
\vspace{-0.5cm}
\hspace{-0cm}
\centering
\begin{tikzpicture}[x=25pt,y=25pt]
\fill[black, fill opacity=0.2] (0,0) .. controls (1,1.3) .. (3,1.5) .. controls (2,0.1) .. (0,0) -- cycle;
\fill[black, fill opacity=0.5] (3,1.5) .. controls (4,2.6) .. (6,3) .. controls (5,2) .. (3,1.5) -- cycle;
\draw (0,0) .. controls (1,1.3) .. (3,1.5) .. controls (5,2) .. (6,3);
\draw[dashed] (0,0) .. controls (2,0.1) .. (3,1.5) .. controls (4,2.6) .. (6,3);
\draw[thick, black, ->] (0, -0.5) -- (0, 3.5);
\draw[thick, black, ->] (-0.5, 0) -- (6, 0) node[anchor=west] {$y$};
\draw[thick, black, ->] (-0.5, 3) node[anchor=east] {$1$} -- (6, 3);
\draw (-0.25,-0.25) node {0};
\end{tikzpicture}
\vspace{-0cm}
\caption{
Illustration of P-CACE and N-CACE. For the conditional CDF $\mathbb{P}(Y<y|X=0,W=w)$ and $\mathbb{P}(Y<y|X=1,W=w)$ shown as the solid and dashed line respectively, P-CACE$(w)$ is given by the light gray region and  N-CACE$(w)$ by the dark gray region.
The difference of the light and dark gray regions gives CACE$(w)$. 
}
\label{fig:1}
\end{wrapfigure}

\subsection{Bounding P-CACE and N-CACE}
The identification result in Theorem~\ref{theo1} relies on the conditional monotonicity assumption, which is implausible in some scenarios.   To address such scenarios, 
we provide the following bounding results:
\begin{restatable}{theorem}{Theoremtwo} [Bounding of P-CACE and N-CACE] \label{theo_b1}
Under SCM ${\cal M}$ and Assumptions \ref{ASEXO2} and \ref{exi1}, for any $w \in \Omega_W$, 
the lower bound of P-CACE is given by Eq.~\eqref{eq14} and the upper bound of P-CACE is given by
\begin{equation}
\label{eq11}
\begin{aligned}
\int_{\Omega_Y}\min\left\{
\begin{array}{c}
1-\mathbb{P}(Y<y|X=1,W=w),\\
\mathbb{P}(Y<y|X=0,W=w)
\end{array}
\right\}dy;
\end{aligned}
\end{equation}
and the lower bound of N-CACE is given by Eq.~\eqref{eq15} and the upper bound of N-CACE is given by
\begin{equation}
\label{eq12}
\begin{aligned}
\int_{\Omega_Y}\min\left\{
\begin{array}{c}
1-\mathbb{P}(Y<y|X=0,W=w),\\
\mathbb{P}(Y<y|X=1,W=w)
\end{array}
\right\}dy.
\end{aligned}
\end{equation}
These are sharp bounds given solely the conditional distribution $\mathbb{P}(Y<y|X=x,W=w)$.
\end{restatable}

%% file: 4_PN-CPICE.tex
In this section, we introduce new causal effect heterogeneity measures for continuous treatment and continuous outcome and provide their identification and bounding theorems. We will utilize stochastic interventions, allowing for the investigation of more flexible scenarios of continuous treatments.

We consider a pair of stochastic interventions
$(X^{\pi_0},X^{\pi_1})$ with joint conditional distribution
$\kappa(dx_0,dx_1\mid w)$ and marginals $\pi_0(dx_0\mid w)$ and
$\pi_1(dx_1\mid w)$. The kernel $\kappa$ specifies their coupling;
independent interventions correspond to $\kappa=\pi_0\otimes\pi_1$,
whereas $X^{\pi_1}=X^{\pi_0}+d$ corresponds to
$\kappa(dx_0,dx_1\mid w)=\pi_0(dx_0\mid w)\delta_{x_0+d}(dx_1)$.
Here, we allow the interventions $(X^{\pi_0},X^{\pi_1})$ to depend on the subjects' covariates $W$. 
The covariate-dependent interventions have been utilized in tailored optimal treatment assignment \citep{Lo2002,Manski2004,Lou2018} and off-policy evaluation and learning \citep{Dudik2014,Oberst2019,Kallus2020,Saito2023}.

Under the stochastic interventions $(X^{\pi_0},X^{\pi_1})$, the ICE is given by $Y_{X^{\pi_1}}-Y_{X^{\pi_0}}$.  
We extend PICE $\mathbb{E}[Y_{X^{\pi_1}}-Y_{X^{\pi_0}}]$ defined by \citet{Munoz2012}, and define the \emph{conditional population intervention causal effect (CPICE)} as follows
{
\begin{equation}
\begin{aligned}
&\text{\normalfont CPICE}(w; \kappa)\defeq\mathbb{E}[Y_{X^{\pi_1}}-Y_{X^{\pi_0}}|W=w]\defeq\int_{\Omega_X}\int_{\Omega_X} \mathbb{E}[Y_{x_1}-Y_{x_0}|W=w] \kappa(dx_0,dx_1\mid w).
\end{aligned}
\end{equation}}
CPICE can be expressed as 
\begin{equation}
\begin{aligned}
\label{eq8}
&\text{\normalfont CPICE}(w; \kappa)=\int_{\Omega_Y} \{\mathbb{P}(Y_{X^{\pi_0}}<y|W=w)-\mathbb{P}(Y_{X^{\pi_1}}<y|W=w)\} dy,
\end{aligned}
\end{equation}where $\mathbb{P}(Y_{X^{\pi_0}}<y\mid W=w)
\defeq
\int_{\Omega_X}
\mathbb{P}(Y_{x_0}<y\mid W=w)\,
\pi_0(dx_0\mid w)$.


\subsection{P-CPICE and N-CPICE for measuring causal effect heterogeneity}

First, we extend the principal stratification in Def.~\ref{def7} from binary treatment to stochastic interventions.
\begin{definition}[Potential outcome types for a continuous outcome under stochastic interventions $(X^{\pi_0},X^{\pi_1})$]
For any $y \in \Omega_Y$, 
we name subjects: 
\begin{enumerate}
\setlength{\parskip}{0cm}
\setlength{\itemsep}{0.2cm}
\item for which $\{Y_{X^{\pi_0}}< y,Y_{X^{\pi_1}}\geq y\}$, {\it the (stochastically) positively affected subjects at point $y$},
\item for which $\{Y_{X^{\pi_0}}\geq y,Y_{X^{\pi_1}}< y\}$, {\it the (stochastically) negatively affected  subjects at point $y$},
\item for which $\{Y_{X^{\pi_0}}\geq y,Y_{X^{\pi_1}}\geq y\}$, {\it the (stochastically) positively immutables at point $y$},
\item for which $\{Y_{X^{\pi_0}}< y,Y_{X^{\pi_1}}< y\}$, {\it the (stochastically) negatively immutables at point $y$}.
\vspace{-0cm}
\end{enumerate}
\end{definition}
\noindent The stochastically positively affected subjects at point $y$ are those propelled by the treatment from an outcome less than $y$ to at least $y$ by changing the stochastic intervention from $X^{\pi_0}$ to $X^{\pi_1}$ under the coupling $\kappa$.
The stochastically negatively affected subjects at point $y$ are those propelled by the treatment from an outcome at least $y$ to less than $y$ by changing the stochastic intervention from $X^{\pi_0}$ to $X^{\pi_1}$ under the coupling $\kappa$.
The stochastically immutables  at point $y$ are those whose outcome remains at least $y$ or
less than $y$ even after changing the intervention from $X^{\pi_0}$ to $X^{\pi_1}$ under the coupling $\kappa$.


We have the following useful relationship, which is an extension of Lemma \ref{LEM1}:
\begin{restatable}{lemma}{Lemmatwo}
\label{LEM3}
For any $y \in \Omega_Y$ and $w \in \Omega_W$,  we have 
\begin{equation}
\begin{aligned}
&\mathbb{P}(Y_{X^{\pi_0}}<y|W=w)-\mathbb{P}(Y_{X^{\pi_1}}<y|W=w)\\
&=\mathbb{P}(Y_{X^{\pi_0}}<y\leq Y_{X^{\pi_1}}|W=w)-\mathbb{P}(Y_{X^{\pi_1}}<y\leq Y_{X^{\pi_0}}|W=w).
\end{aligned}
\end{equation}
\end{restatable}



We introduce the following measures to quantify the positive and absolute negative contributions of individual causal effects to CPICE.
\begin{definition}[P-CPICE and N-CPICE]
For any $w\in\Omega_W$, we define the positive and absolute negative contributions of individual causal effects to CPICE under the paired stochastic interventions $(X^{\pi_0},X^{\pi_1})$, denoted by P-CPICE and N-CPICE, respectively, as
\begin{equation}
\begin{aligned}
&\text{\normalfont P-CPICE}(w,\kappa)\defeq\int_{\Omega_Y}\mathbb{P}(Y_{X^{\pi_0}} < y \leq Y_{X^{\pi_1}}|W=w)dy,
\end{aligned}
\end{equation}
\begin{equation}
\begin{aligned}
&\text{\normalfont N-CPICE}(w,\kappa)\defeq\int_{\Omega_Y}\mathbb{P}(Y_{X^{\pi_1}} < y \leq Y_{X^{\pi_0}}|W=w)dy.
\end{aligned}
\end{equation}
We define {$\text{\normalfont P-PICE}(\kappa)\defeq\mathbb{E}_W[\text{\normalfont P-CPICE}(W,\kappa)]$}, and {$\text{\normalfont N-PICE}(\kappa)\defeq\mathbb{E}_W[\text{\normalfont N-CPICE}(W,\kappa)]$}.
\end{definition}
P-CPICE and N-CPICE can also be expressed as 
\begin{equation}
\begin{aligned}
&\text{\normalfont P-CPICE}(w,\kappa)=\mathbb{E}\left[\int_{\Omega_Y} \mathbb{I}(Y_{X^{\pi_0}} < y \leq Y_{X^{\pi_1}})dy\Big|W=w\right],
\end{aligned}
\end{equation}
\begin{equation}
\begin{aligned}
&\text{\normalfont N-CPICE}(w,\kappa)=\mathbb{E}\left[\int_{\Omega_Y} \mathbb{I}(Y_{X^{\pi_1}} < y \leq Y_{X^{\pi_0}})dy\Big|W=w\right].
\end{aligned}
\end{equation}
We have
\begin{align}\label{eq-bbb}
 \int_{\Omega_Y} \mathbb{I}(Y_{X^{\pi_0}} < y \leq Y_{X^{\pi_1}})dy = \begin{cases}
   \text{\normalfont ICE} & \text{\normalfont ICE} > 0 \\
   0 & \text{\normalfont ICE} \leq 0
 \end{cases},
\end{align}
where $\text{\normalfont ICE}=Y_{X^{\pi_1}}-Y_{X^{\pi_0}}$. 
From Eq.~(\ref{eq-bbb}), P-CPICE only accounts for individuals with positive ICE, that is, P-CPICE measures the portion of the conditional average causal effects (CPICE) over positively affected subjects (at some point). Similarly, N-CPICE can be interpreted as measuring (the absolute value of) the negative portion of  $\text{\normalfont CPICE}$ over negatively affected subjects (at some point). Both P-CPICE and N-CPICE take non-negative values.

Importantly, CPICE can be decomposed into P-CPICE and N-CPICE as follows.
\begin{assumption}[Finiteness of P-CPICE and N-CPICE]
\label{exi3}
We assume $\int_{\Omega_Y}\mathbb{P}(Y_{X^{\pi_0}} < y \leq Y_{X^{\pi_1}}|W=w)dy<\infty$ and $\int_{\Omega_Y}\mathbb{P}(Y_{X^{\pi_1}} < y \leq Y_{X^{\pi_0}}|W=w)dy<\infty$
for any $w \in \Omega_W$.
\end{assumption}
\begin{restatable}{proposition}{Propositiontwo}[Decomposition]
\label{PROP3}
Under SCM ${\cal M}$ and Assumption \ref{exi3},
for any $w \in \Omega_W$, we have
\begin{equation}
\begin{aligned}
&\text{\normalfont CPICE}(w,\kappa)=\text{\normalfont P-CPICE}(w,\kappa)-\text{\normalfont N-CPICE}(w,\kappa).
\end{aligned}
\end{equation}
\end{restatable}


{We revisit Examples 1-3 in the previous section to illustrate the P-PICE and N-PICE measures. We consider stochastic interventions $X^{\pi_0}=X\sim{\cal N}(0,1)$ and $X^{\pi_1}=X+d$, which increase the treatment of all subjects by a fixed amount $d>0$ from their current values.}

\paragraph{Example 1' (Homogeneous and positive ICE).}
Consider the model $Y:=X+U_Y$, $U_Y \sim {\cal N}(0,1)$. 
We have $\text{\normalfont PICE}=\mathbb{E}[Y_{X^{\pi_1}}]-\mathbb{E}[Y_{X^{\pi_0}}]=\mathbb{E}[X+d+U_Y]-\mathbb{E}[X+U_Y]=d$,
$\text{\normalfont P-PICE}
=\int_{-\infty}^{\infty} \mathbb{P}(X+U_Y < y \leq X+d+U_Y)dy=d$, and
$\text{\normalfont N-PICE}
=\int_{-\infty}^{\infty} \mathbb{P}(X+d+U_Y < y \leq X+U_Y)dy=0$.
These results are consistent with the fact that no subjects have negative individual causal effects in this setting.

\paragraph{Example 2' (Heterogeneous ICE).}
We consider the model $Y:=XU_Y$, $U_Y \sim {\cal N}(0,1)$.
Under this setting, 
$\text{\normalfont PICE}=\mathbb{E}[Y_{X^{\pi_1}}]-\mathbb{E}[Y_{X^{\pi_0}}]=\mathbb{E}[(X+d)U_Y]-\mathbb{E}[XU_Y]=\mathbb{E}[d U_Y]=0$. We can show that $\text{\normalfont P-PICE}=\text{\normalfont N-PICE} >0$. 
Even when PICE is zero, P-PICE and N-PICE can reveal offsetting positive and negative contributions of individual causal effects that cancel in the average.

\paragraph{Example 3' (Homogeneous and zero ICE).}
We consider the model $Y:=U_Y$, $U_Y \sim {\cal N}(0,1)$, where all subjects have zero ICE.
Under this setting,
$\text{\normalfont PICE}=\mathbb{E}[Y_{X^{\pi_1}}]-\mathbb{E}[Y_{X^{\pi_0}}]=0$,
$\text{\normalfont P-PICE}
=\int_{-\infty}^{\infty} \mathbb{P}(U_Y < y \leq U_Y)dy=0$, and
$\text{\normalfont N-PICE}
=\int_{-\infty}^{\infty} \mathbb{P}(U_Y < y \leq U_Y)dy=0$.

PICEs in both Examples 2' and 3' are equal to 0; however, 
they have completely different patterns of causal effect heterogeneity.
{Our newly introduced P-PICE and N-PICE are able to measure  effect heterogeneity unexplained by PICE.}

\subsection{Identification of P-CPICE and N-CPICE}

We have obtained the following identification theorem:
\begin{restatable}{theorem}{Theoremfour}[Identification of P-CPICE and N-CPICE]
\label{theo3}
Under SCM ${\cal M}$ and Assumptions \ref{ASEXO2}, \ref{MONO2}, and \ref{exi3}, 
P-CPICE and N-CPICE are identifiable by
{
\begin{equation}
\begin{aligned}
\label{eq36}
&\text{\normalfont P-CPICE}(w,\kappa)=\int_{\Omega_X}\int_{\Omega_X}\int_{\Omega_Y}\max\Big\{\mathbb{P}(Y<y|X=x_0,W=w)\\
&\hspace{1cm}-\mathbb{P}(Y<y|X=x_1,W=w),0\Big\} dy\kappa(dx_0,dx_1\mid w),
\end{aligned} 
\end{equation}
\begin{equation}
\begin{aligned}
\label{eq37}
&\text{\normalfont N-CPICE}(w,\kappa)=\int_{\Omega_X}\int_{\Omega_X}\int_{\Omega_Y} \max\Big\{\mathbb{P}(Y<y|X=x_1,W=w)\\
&\hspace{1cm}-\mathbb{P}(Y<y|X=x_0,W=w),0\Big\}dy\kappa(dx_0,dx_1\mid w).
\end{aligned} 
\end{equation}
}
\end{restatable}
\noindent When $\pi_0(dx_0\mid w)=\delta_0(dx_0)$ and
$\pi_1(dx_1\mid w)=\delta_1(dx_1)$, where $\delta_a$
denotes the Dirac measure at $a$, Theorem~\ref{theo3}
reduces to Theorem~\ref{theo1}.



\subsection{Bounding P-CPICE and N-CPICE}
To deal with scenarios where the conditional monotonicity (Assumption \ref{MONO2}) is deemed implausible, we provide the following bounding results:
\begin{restatable}{theorem}{Theoremfive}\label{theo_b2}
Under SCM ${\cal M}$ and Assumptions \ref{ASEXO2} and \ref{exi3}, for any $w \in \Omega_W$,
the lower bound of P-CPICE is given by Eq.~\eqref{eq36} and the upper bound of P-CPICE is given by
\begin{equation}
\label{eq21}
\begin{aligned}
&\int_{\Omega_X}\int_{\Omega_X}\int_{\Omega_Y} \min\left\{
\begin{array}{c}
1-\mathbb{P}(Y<y|X=x_1,W=w),\\
\mathbb{P}(Y<y|X=x_0,W=w)
\end{array}\right\}dy\kappa(dx_0,dx_1\mid w);
\end{aligned}
\end{equation}
and the lower bound of N-CPICE is given by Eq.~\eqref{eq37}
and the upper bound of N-CPICE is given by
\begin{equation}
\label{eq22}
\begin{aligned}
&\int_{\Omega_X}\int_{\Omega_X}\int_{\Omega_Y} \min\left\{
\begin{array}{c}
1-\mathbb{P}(Y<y|X=x_0,W=w),\\
\mathbb{P}(Y<y|X=x_1,W=w)
\end{array}
\right\}dy\kappa(dx_0,dx_1\mid w).
\end{aligned}
\end{equation}
\end{restatable}
The lower bounds are sharp. The upper bounds are obtained by
integrating pointwise sharp Fr\'echet bounds, but their
simultaneous sharpness is not claimed.

%% file: 7_Application.tex


\paragraph{Estimation.}
We estimate P-CACE, N-CACE, P-CPICE, and N-CPICE 
from finite samples by first estimating conditional CDFs and then estimating  the integrations by Monte Carlo integration. 
In principle, any conditional CDF estimators can be used. In the experiments, we use the local linear (LL) estimator \citep{Hall1999}.
The details of the estimation methods are provided  in Appendix~\ref{appA}. Numerical experiments are provided in Appendix~\ref{sec6} to illustrate the finite-sample properties of the estimators.
The experiment results show that the estimates approach the ground truth and their uncertainty decreases as the sample size grows, that the estimated bounds contain the ground truth even when Assumption~\ref{MONO2} is violated, and that Lasso regularization stabilizes the estimators with high-dimensional covariates, although their bias grows with the dimension.
We report each estimate with a 95\% bootstrap confidence interval based on 1000 replicates \citep{Efron1979}.


\paragraph{Dataset.}
We take up an open dataset about medical costs in \citep{Lantz2019}
(\url{https://github.com/stedy/Machine-Learning-with-R-datasets/blob/master/insurance.csv}). 
This dataset contains 1,338 examples of beneficiaries currently enrolled in an insurance plan, with characteristics of patients, 
as well as the total medical expenses billed to the plan for the calendar year.
\citet{Lantz2019} aimed to predict medical expenses based on patients' characteristics.

\paragraph{Analysis.}
We study the causal impact of the patients' Body Mass Index (BMI) ($X$) on their medical expenses ($Y$ in $\$$) billed by health insurance  and its heterogeneity.
We consider patients who are 30 years old, non-smokers, men, and have 1 child, and denote these characteristics as covariates $W=w$. We match exactly on sex, smoking status, and number of children ($n=130$) and apply the LL estimator to BMI and age.

We assume that Assumption \ref{ASEXO2} (Conditional exogeneity) is reasonable after adjusting for potential confounders $W$.  
Assumption \ref{MONO2} (Conditional monotonicity) requires that, within each level of $W$, the ordering of an individual's potential medical expenditures under different BMI levels does not reverse. This kind of “non-crossing” can be plausible in settings where individuals tend to maintain their relative position in the cost distribution, regardless of moderate changes in BMI. 
For example, individuals with persistently high medical expenditures—often due to chronic or severe health conditions—are likely to remain high-cost individuals whether their BMI is 20 or 40. Similarly, individuals who seldom use medical care and therefore incur consistently low expenditures are likely to remain relatively low-cost across BMI levels. In such situations, the ordering of potential outcomes across different BMI values remains stable, making the conditional monotonicity assumption more reasonable. 

\begin{wraptable}{r}{0.57\textwidth}
\vspace{-2em}
\centering
\footnotesize
\setlength{\tabcolsep}{2pt}
\caption{
Estimates with 95\% CI in US\$. The estimates under Assumption~\ref{MONO2} are also the lower bounds without it.
}
\label{tab:app}
\begin{tabular}{lrr}
\toprule
Measure & Estimate under Assumption~\ref{MONO2} & Upper bound \\
\midrule
CACE    & 3,955 [$-$444, 11,748] & -- \\
P-CACE  & 4,068 [367, 11,748]    & 5,019 [614, 11,766] \\
N-CACE  & 112 [0, 913]           & 1,063 [117, 1,296] \\
\midrule
CPICE   & 357 [$-$235, 1,095]    & -- \\
P-CPICE & 508 [183, 1,363]       & 2,445 [1,001, 3,711] \\
N-CPICE & 151 [74, 673]          & 2,087 [931, 3,024] \\
\bottomrule
\end{tabular}
\vspace{-2em}
\end{wraptable}

\paragraph{Results.}

Table~\ref{tab:app} reports the estimates and their 95\% confidence intervals (CIs). 
The positive CACE (\$3{,}955) estimate suggests that, \emph{on average}, a change in BMI from $20$ to $40$ increases medical costs for this subpopulation, although the CI is wide and includes zero.
Next, to examine how positive and negative individual causal effects contribute to CACE, we evaluate P-CACE and N-CACE.
N-CACE (\$112) is small relative to P-CACE (\$4{,}068), indicating that negative individual causal effects make a small aggregate contribution compared with positive individual causal effects in this analysis.
The upper bound of N-CACE (\$1{,}063) is also small relative to P-CACE (at least \$4{,}068), meaning that even in the worst case, without Assumption~\ref{MONO2} (Conditional monotonicity), negative individual causal effects contribute less than positive ones.

Researchers may consider the causal effect heterogeneity of the populational changes of BMI, e.g., increasing (or decreasing) BMI by 3 for all subjects.
This question can be answered by CPICE under stochastic interventions $(X^{\pi_0},X^{\pi_1})=(X,X+3)$.
The positive CPICE estimate (\$357) suggests that, \emph{on average}, increasing BMI by 3 increases medical costs for this subpopulation, although the CI is wide and includes zero.
To examine how positive and negative individual causal effects contribute to the CPICE, we estimate P-CPICE and N-CPICE.
The P-CPICE (\$508) and N-CPICE (\$151) estimates suggest that both positive and negative individual causal effects may contribute to the average effect.
Also, the upper bounds of P-CPICE (\$2{,}445) and N-CPICE (\$2{,}087) are similar, meaning that in the worst case, without Assumption~\ref{MONO2} (Conditional monotonicity), large positive and negative contributions could nearly cancel each other, leaving a CPICE of only \$357.

This heterogeneity may reflect unobserved effect modifiers, that is, unmeasured factors that change the size or direction of the effect across individuals, such as income or comorbidities.
Measuring these factors could help explain the remaining heterogeneity, which motivates researchers to collect additional information.

%% file: 8_Conclusion.tex
CACE is the primary measure researchers use to assess causal effect heterogeneity. 
CACE summarizes how average effects vary with observed covariates, but it can obscure positive and negative individual effects within covariate-defined subpopulations. This paper provides measures that separate these contributions for continuous outcomes under both binary and continuous treatments. 
We establish identification theorems for these measures. 
A limitation of these identification results is their reliance on the monotonicity assumption, which may be restrictive in some real-world applications. To address this, we derive bounds for settings where the monotonicity is deemed implausible. These measures, together with their identification and bounding methods,  
offer valuable tools for researchers seeking a more nuanced understanding of causal effect heterogeneity, offering a more complete picture of treatment effect variation. 

The proposed measures may inform practical decision-making by revealing negative contributions that a positive average effect can mask. 
Even when CACE or CPICE is positive, a large N-CACE or N-CPICE indicates a substantial aggregate negative contribution and may warrant closer scrutiny of the treatment. Conversely, a small value may offer reassurance about the aggregate magnitude of negative effects.  
By complementing CACE and CPICE, these measures support a more nuanced assessment of treatments and their consequences for individuals.

Many sophisticated algorithms for estimating CACE have
been developed, such as causal forest \citep{Athey2016}, causal boosting \citep{Powers2018}, double machine
learning \citep{Chernozhukov2018}, meta-learners
\citep{Kunzel2019,Nie2020,Kennedy2023},
and CATE Lasso \citep{Kato2023}. 
Developing more sophisticated estimation methods for our measures with improved statistical or computational properties is a direction for future research.\\


%% file: A0_Table.tex
We provide a summary table comparing the applicable settings of our  causal effect heterogeneity measures to existing ones 
in Table \ref{tab:app1}.

\begin{table}[h]
\centering
\caption{Summary table of causal effect heterogeneity measures}
\hspace{-0cm}
\label{tab:app1}
\vspace{0cm}
\scalebox{0.7}{
\begin{tabular}{l|cccc}
\hline
Measures & Binary $X$ and $Y$ & Binary $X$, Continuous $Y$ & Continuous $X$ and $Y$ & Heterogeneity\\
\hline \hline
CACE & $\checkmark$ & $\checkmark$ & $\times$ & Covariates \\
CPICE & $\checkmark$ & $\checkmark$ & $\checkmark$ & Covariates \\
TBR & $\checkmark$ & $\times$ & $\times$ & {Proportion of the useful treatment group} \\
THR = FNA & $\checkmark$ & $\times$ & $\times$ & {Proportion of the harmful treatment group} \\
TBR & $\checkmark$ & $\checkmark$ {Binarizing $Y$}& $\times$ & 
\begin{tabular}{c}
{depending on thresholds $y_0$ and $y_1$}
\end{tabular}\\
THR & $\checkmark$ & $\checkmark$ {Binarizing $Y$}& $\times$ &
\begin{tabular}{c}
{depending on thresholds $y_0$ and $y_1$}
\end{tabular}\\
$\text{\normalfont TBR}_c(w)$ & $\checkmark$ & $\checkmark$ & $\times$ &\begin{tabular}{c}
{depending on threshold $c$}
\end{tabular}\\
$\text{\normalfont THR}_c(w)$ & $\checkmark$ & $\checkmark$ & $\times$ &\begin{tabular}{c}
{depending on threshold $c$}
\end{tabular}\\
{{\small $\mathbb{P}(\zeta(W)(Y_1-Y_0)<\delta)$}} & $\checkmark$ & $\checkmark$ & $\times$ &
{Proportion with weighted ICE below $\delta$}\\
P-CACE (ours) & $\checkmark$ & $\checkmark$ & $\times$ & Positive contribution to CACE \\
N-CACE (ours) & $\checkmark$ & $\checkmark$ & $\times$ & Absolute negative contribution to CACE \\
P-CPICE (ours)& $\checkmark$ & $\checkmark$ & $\checkmark$ & Positive contribution to CPICE \\
N-CPICE (ours)& $\checkmark$ & $\checkmark$ & $\checkmark$ & Absolute negative contribution to CPICE\\
\hline
\end{tabular}
}
\end{table}

%% file: App_Example.tex
{
\noindent{\bf Example 4 (Heterogeneous ICE with covariates).}
We consider the model:
\begin{align}
\label{eq119}
Y:=X+\mathbb{I}(X=1)U_Y+\mathbb{I}(X=0)WU_Y , U_Y \sim {\cal N}(0,1).
\end{align}
Under this setting, 
CACE is always $1$ for any $w \in \Omega_W$.
However, P-CACE and N-CACE depend on the value of $W$.
We plot the CACE, P-CACE, and N-CACE in Figure \ref{fig:d1}.
Near $W=0$, P-CACE is almost 1 and N-CACE is almost 0.
As the value of $W$ increases, the values of both P-CACE and N-CACE increase while the difference of P-CACE and N-CACE is always 1.
This synthetic example illustrates that P-CACE and N-CACE can reveal aspect of causal heterogeneity that CACE cannot capture.
}

\noindent\textbf{Remark: Comparison with $\text{\normalfont THR}_c$ and $\text{\normalfont TBR}_c$.} 
\citet{Yin2018} introduced heterogeneous measures for a continuous outcome, $\text{\normalfont TBR}_c(w)\defeq\mathbb{P}(Y_1-Y_0>c|W=w)$ and $\text{\normalfont THR}_c(w)\defeq\mathbb{P}(Y_0-Y_1>c|W=w)$ given a constant $c$, and showed the decomposition $\displaystyle\text{\normalfont CACE}(w)=\int_{0}^{\infty}\{\text{\normalfont TBR}_c(w)-\text{\normalfont THR}_c(w)\}dc$.
If both integrals are finite, this decomposition coincides with ours.
\begin{equation}
\mathrm{P\text{-}CACE}(w)=\int_0^\infty\mathrm{TBR}_c(w)\,dc,\qquad
\mathrm{N\text{-}CACE}(w)=\int_0^\infty\mathrm{THR}_c(w)\,dc.
\end{equation}


\input{Fig_D1}

%% file: Fig_D1.tex
\begin{figure}[ht]
\vspace{-0cm}
    \centering
\begin{tikzpicture}[x=1pt,y=1pt]
\definecolor{fillColor}{RGB}{255,255,255}
\path[use as bounding box,fill=fillColor,fill opacity=0.00] (0,0) rectangle (252.94,216.81);
\begin{scope}
\path[clip] ( 49.20, 61.20) rectangle (227.75,167.61);
\definecolor{drawColor}{RGB}{0,0,0}

\path[draw=drawColor,line width= 0.4pt,line join=round,line cap=round] ( 55.81, 97.98) --
	( 57.48, 97.98) --
	( 59.15, 97.98) --
	( 60.82, 97.98) --
	( 62.49, 97.98) --
	( 64.16, 97.98) --
	( 65.83, 97.98) --
	( 67.50, 97.98) --
	( 69.17, 97.98) --
	( 70.84, 97.98) --
	( 72.51, 97.98) --
	( 74.18, 97.98) --
	( 75.85, 97.98) --
	( 77.52, 97.98) --
	( 79.19, 97.98) --
	( 80.86, 97.98) --
	( 82.53, 97.98) --
	( 84.20, 97.98) --
	( 85.87, 97.98) --
	( 87.54, 97.98) --
	( 89.21, 97.98) --
	( 90.88, 97.98) --
	( 92.55, 97.98) --
	( 94.22, 97.98) --
	( 95.89, 97.98) --
	( 97.56, 97.98) --
	( 99.23, 97.98) --
	(100.90, 97.98) --
	(102.57, 97.98) --
	(104.24, 97.98) --
	(105.91, 97.98) --
	(107.58, 97.98) --
	(109.25, 97.98) --
	(110.92, 97.98) --
	(112.59, 97.98) --
	(114.26, 97.98) --
	(115.93, 97.98) --
	(117.60, 97.98) --
	(119.27, 97.98) --
	(120.94, 97.98) --
	(122.61, 97.98) --
	(124.28, 97.98) --
	(125.95, 97.98) --
	(127.62, 97.98) --
	(129.29, 97.98) --
	(130.96, 97.98) --
	(132.63, 97.98) --
	(134.30, 97.98) --
	(135.97, 97.98) --
	(137.64, 97.98) --
	(139.31, 97.98) --
	(140.98, 97.98) --
	(142.65, 97.98) --
	(144.32, 97.98) --
	(145.99, 97.98) --
	(147.66, 97.98) --
	(149.33, 97.98) --
	(151.00, 97.98) --
	(152.67, 97.98) --
	(154.34, 97.98) --
	(156.01, 97.98) --
	(157.68, 97.98) --
	(159.35, 97.98) --
	(161.02, 97.98) --
	(162.69, 97.98) --
	(164.36, 97.98) --
	(166.03, 97.98) --
	(167.70, 97.98) --
	(169.37, 97.98) --
	(171.04, 97.98) --
	(172.71, 97.98) --
	(174.38, 97.98) --
	(176.05, 97.98) --
	(177.71, 97.98) --
	(179.38, 97.98) --
	(181.05, 97.98) --
	(182.72, 97.98) --
	(184.39, 97.98) --
	(186.06, 97.98) --
	(187.73, 97.98) --
	(189.40, 97.98) --
	(191.07, 97.98) --
	(192.74, 97.98) --
	(194.41, 97.98) --
	(196.08, 97.98) --
	(197.75, 97.98) --
	(199.42, 97.98) --
	(201.09, 97.98) --
	(202.76, 97.98) --
	(204.43, 97.98) --
	(206.10, 97.98) --
	(207.77, 97.98) --
	(209.44, 97.98) --
	(211.11, 97.98) --
	(212.78, 97.98) --
	(214.45, 97.98) --
	(216.12, 97.98) --
	(217.79, 97.98) --
	(219.46, 97.98) --
	(221.13, 97.98);
\end{scope}
\begin{scope}
\path[clip] (  0.00,  0.00) rectangle (252.94,216.81);
\definecolor{drawColor}{RGB}{0,0,0}

\path[draw=drawColor,line width= 0.4pt,line join=round,line cap=round] ( 54.14, 61.20) -- (221.13, 61.20);

\path[draw=drawColor,line width= 0.4pt,line join=round,line cap=round] ( 54.14, 61.20) -- ( 54.14, 55.20);

\path[draw=drawColor,line width= 0.4pt,line join=round,line cap=round] ( 87.54, 61.20) -- ( 87.54, 55.20);

\path[draw=drawColor,line width= 0.4pt,line join=round,line cap=round] (120.94, 61.20) -- (120.94, 55.20);

\path[draw=drawColor,line width= 0.4pt,line join=round,line cap=round] (154.34, 61.20) -- (154.34, 55.20);

\path[draw=drawColor,line width= 0.4pt,line join=round,line cap=round] (187.73, 61.20) -- (187.73, 55.20);

\path[draw=drawColor,line width= 0.4pt,line join=round,line cap=round] (221.13, 61.20) -- (221.13, 55.20);

\node[text=drawColor,anchor=base,inner sep=0pt, outer sep=0pt, scale=  1.00] at ( 54.14, 39.60) {0};

\node[text=drawColor,anchor=base,inner sep=0pt, outer sep=0pt, scale=  1.00] at ( 87.54, 39.60) {2};

\node[text=drawColor,anchor=base,inner sep=0pt, outer sep=0pt, scale=  1.00] at (120.94, 39.60) {4};

\node[text=drawColor,anchor=base,inner sep=0pt, outer sep=0pt, scale=  1.00] at (154.34, 39.60) {6};

\node[text=drawColor,anchor=base,inner sep=0pt, outer sep=0pt, scale=  1.00] at (187.73, 39.60) {8};

\node[text=drawColor,anchor=base,inner sep=0pt, outer sep=0pt, scale=  1.00] at (221.13, 39.60) {10};

\path[draw=drawColor,line width= 0.4pt,line join=round,line cap=round] ( 49.20, 65.14) -- ( 49.20,163.67);

\path[draw=drawColor,line width= 0.4pt,line join=round,line cap=round] ( 49.20, 65.14) -- ( 43.20, 65.14);

\path[draw=drawColor,line width= 0.4pt,line join=round,line cap=round] ( 49.20, 81.56) -- ( 43.20, 81.56);

\path[draw=drawColor,line width= 0.4pt,line join=round,line cap=round] ( 49.20, 97.98) -- ( 43.20, 97.98);

\path[draw=drawColor,line width= 0.4pt,line join=round,line cap=round] ( 49.20,114.40) -- ( 43.20,114.40);

\path[draw=drawColor,line width= 0.4pt,line join=round,line cap=round] ( 49.20,130.83) -- ( 43.20,130.83);

\path[draw=drawColor,line width= 0.4pt,line join=round,line cap=round] ( 49.20,147.25) -- ( 43.20,147.25);

\path[draw=drawColor,line width= 0.4pt,line join=round,line cap=round] ( 49.20,163.67) -- ( 43.20,163.67);

\node[text=drawColor,rotate= 90.00,anchor=base,inner sep=0pt, outer sep=0pt, scale=  1.00] at ( 34.80, 65.14) {-1};

\node[text=drawColor,rotate= 90.00,anchor=base,inner sep=0pt, outer sep=0pt, scale=  1.00] at ( 34.80, 81.56) {0};

\node[text=drawColor,rotate= 90.00,anchor=base,inner sep=0pt, outer sep=0pt, scale=  1.00] at ( 34.80, 97.98) {1};

\node[text=drawColor,rotate= 90.00,anchor=base,inner sep=0pt, outer sep=0pt, scale=  1.00] at ( 34.80,114.40) {2};

\node[text=drawColor,rotate= 90.00,anchor=base,inner sep=0pt, outer sep=0pt, scale=  1.00] at ( 34.80,130.83) {3};

\node[text=drawColor,rotate= 90.00,anchor=base,inner sep=0pt, outer sep=0pt, scale=  1.00] at ( 34.80,147.25) {4};

\node[text=drawColor,rotate= 90.00,anchor=base,inner sep=0pt, outer sep=0pt, scale=  1.00] at ( 34.80,163.67) {5};

\path[draw=drawColor,line width= 0.4pt,line join=round,line cap=round] ( 49.20, 61.20) --
	(227.75, 61.20) --
	(227.75,167.61) --
	( 49.20,167.61) --
	cycle;
\end{scope}
\begin{scope}
\path[clip] ( 49.20, 61.20) rectangle (227.75,167.61);
\definecolor{drawColor}{RGB}{0,0,0}

\path[draw=drawColor,line width= 0.4pt,dash pattern=on 4pt off 4pt ,line join=round,line cap=round] ( 55.81, 99.06) --
	( 57.48, 98.70) --
	( 59.15, 98.43) --
	( 60.82, 98.09) --
	( 62.49, 98.14) --
	( 64.16, 97.91) --
	( 65.83, 97.97) --
	( 67.50, 97.91) --
	( 69.17, 97.97) --
	( 70.84, 97.99) --
	( 72.51, 97.93) --
	( 74.18, 97.99) --
	( 75.85, 98.02) --
	( 77.52, 98.00) --
	( 79.19, 98.13) --
	( 80.86, 98.12) --
	( 82.53, 98.37) --
	( 84.20, 98.64) --
	( 85.87, 99.06) --
	( 87.54, 99.39) --
	( 89.21, 99.81) --
	( 90.88,100.06) --
	( 92.55,100.66) --
	( 94.22,101.18) --
	( 95.89,101.69) --
	( 97.56,102.28) --
	( 99.23,102.82) --
	(100.90,103.49) --
	(102.57,103.96) --
	(104.24,104.55) --
	(105.91,105.08) --
	(107.58,105.67) --
	(109.25,106.45) --
	(110.92,106.90) --
	(112.59,107.41) --
	(114.26,107.99) --
	(115.93,108.63) --
	(117.60,109.34) --
	(119.27,109.94) --
	(120.94,110.52) --
	(122.61,111.07) --
	(124.28,111.81) --
	(125.95,112.52) --
	(127.62,112.94) --
	(129.29,113.61) --
	(130.96,114.27) --
	(132.63,114.72) --
	(134.30,115.62) --
	(135.97,116.16) --
	(137.64,116.76) --
	(139.31,117.42) --
	(140.98,118.12) --
	(142.65,118.61) --
	(144.32,119.42) --
	(145.99,120.31) --
	(147.66,120.70) --
	(149.33,121.17) --
	(151.00,121.93) --
	(152.67,122.58) --
	(154.34,123.28) --
	(156.01,123.85) --
	(157.68,124.68) --
	(159.35,125.05) --
	(161.02,125.69) --
	(162.69,126.35) --
	(164.36,127.15) --
	(166.03,127.65) --
	(167.70,128.56) --
	(169.37,128.82) --
	(171.04,129.44) --
	(172.71,130.37) --
	(174.38,130.65) --
	(176.05,131.39) --
	(177.71,132.31) --
	(179.38,132.93) --
	(181.05,133.54) --
	(182.72,134.30) --
	(184.39,134.69) --
	(186.06,135.58) --
	(187.73,136.04) --
	(189.40,136.88) --
	(191.07,137.34) --
	(192.74,138.16) --
	(194.41,138.60) --
	(196.08,139.58) --
	(197.75,139.96) --
	(199.42,140.89) --
	(201.09,141.38) --
	(202.76,141.99) --
	(204.43,142.57) --
	(206.10,143.09) --
	(207.77,144.15) --
	(209.44,144.44) --
	(211.11,145.20) --
	(212.78,145.90) --
	(214.45,146.41) --
	(216.12,147.06) --
	(217.79,147.77) --
	(219.46,148.22) --
	(221.13,149.24);
\end{scope}
\begin{scope}
\path[clip] (  0.00,  0.00) rectangle (252.94,216.81);
\definecolor{drawColor}{RGB}{0,0,0}

\path[draw=drawColor,line width= 0.4pt,line join=round,line cap=round] ( 54.14, 61.20) -- (221.13, 61.20);

\path[draw=drawColor,line width= 0.4pt,line join=round,line cap=round] ( 54.14, 61.20) -- ( 54.14, 55.20);

\path[draw=drawColor,line width= 0.4pt,line join=round,line cap=round] ( 87.54, 61.20) -- ( 87.54, 55.20);

\path[draw=drawColor,line width= 0.4pt,line join=round,line cap=round] (120.94, 61.20) -- (120.94, 55.20);

\path[draw=drawColor,line width= 0.4pt,line join=round,line cap=round] (154.34, 61.20) -- (154.34, 55.20);

\path[draw=drawColor,line width= 0.4pt,line join=round,line cap=round] (187.73, 61.20) -- (187.73, 55.20);

\path[draw=drawColor,line width= 0.4pt,line join=round,line cap=round] (221.13, 61.20) -- (221.13, 55.20);

\node[text=drawColor,anchor=base,inner sep=0pt, outer sep=0pt, scale=  1.00] at ( 54.14, 39.60) {0};

\node[text=drawColor,anchor=base,inner sep=0pt, outer sep=0pt, scale=  1.00] at ( 87.54, 39.60) {2};

\node[text=drawColor,anchor=base,inner sep=0pt, outer sep=0pt, scale=  1.00] at (120.94, 39.60) {4};

\node[text=drawColor,anchor=base,inner sep=0pt, outer sep=0pt, scale=  1.00] at (154.34, 39.60) {6};

\node[text=drawColor,anchor=base,inner sep=0pt, outer sep=0pt, scale=  1.00] at (187.73, 39.60) {8};

\node[text=drawColor,anchor=base,inner sep=0pt, outer sep=0pt, scale=  1.00] at (221.13, 39.60) {10};

\path[draw=drawColor,line width= 0.4pt,line join=round,line cap=round] ( 49.20, 65.14) -- ( 49.20,163.67);

\path[draw=drawColor,line width= 0.4pt,line join=round,line cap=round] ( 49.20, 65.14) -- ( 43.20, 65.14);

\path[draw=drawColor,line width= 0.4pt,line join=round,line cap=round] ( 49.20, 81.56) -- ( 43.20, 81.56);

\path[draw=drawColor,line width= 0.4pt,line join=round,line cap=round] ( 49.20, 97.98) -- ( 43.20, 97.98);

\path[draw=drawColor,line width= 0.4pt,line join=round,line cap=round] ( 49.20,114.40) -- ( 43.20,114.40);

\path[draw=drawColor,line width= 0.4pt,line join=round,line cap=round] ( 49.20,130.83) -- ( 43.20,130.83);

\path[draw=drawColor,line width= 0.4pt,line join=round,line cap=round] ( 49.20,147.25) -- ( 43.20,147.25);

\path[draw=drawColor,line width= 0.4pt,line join=round,line cap=round] ( 49.20,163.67) -- ( 43.20,163.67);

\node[text=drawColor,rotate= 90.00,anchor=base,inner sep=0pt, outer sep=0pt, scale=  1.00] at ( 34.80, 65.14) {-1};

\node[text=drawColor,rotate= 90.00,anchor=base,inner sep=0pt, outer sep=0pt, scale=  1.00] at ( 34.80, 81.56) {0};

\node[text=drawColor,rotate= 90.00,anchor=base,inner sep=0pt, outer sep=0pt, scale=  1.00] at ( 34.80, 97.98) {1};

\node[text=drawColor,rotate= 90.00,anchor=base,inner sep=0pt, outer sep=0pt, scale=  1.00] at ( 34.80,114.40) {2};

\node[text=drawColor,rotate= 90.00,anchor=base,inner sep=0pt, outer sep=0pt, scale=  1.00] at ( 34.80,130.83) {3};

\node[text=drawColor,rotate= 90.00,anchor=base,inner sep=0pt, outer sep=0pt, scale=  1.00] at ( 34.80,147.25) {4};

\node[text=drawColor,rotate= 90.00,anchor=base,inner sep=0pt, outer sep=0pt, scale=  1.00] at ( 34.80,163.67) {5};

\path[draw=drawColor,line width= 0.4pt,line join=round,line cap=round] ( 49.20, 61.20) --
	(227.75, 61.20) --
	(227.75,167.61) --
	( 49.20,167.61) --
	cycle;
\end{scope}
\begin{scope}
\path[clip] ( 49.20, 61.20) rectangle (227.75,167.61);
\definecolor{drawColor}{RGB}{0,0,0}

\path[draw=drawColor,line width= 0.4pt,dash pattern=on 1pt off 3pt ,line join=round,line cap=round] ( 55.81, 82.56) --
	( 57.48, 82.23) --
	( 59.15, 81.95) --
	( 60.82, 81.76) --
	( 62.49, 81.63) --
	( 64.16, 81.58) --
	( 65.83, 81.56) --
	( 67.50, 81.56) --
	( 69.17, 81.56) --
	( 70.84, 81.56) --
	( 72.51, 81.56) --
	( 74.18, 81.56) --
	( 75.85, 81.56) --
	( 77.52, 81.58) --
	( 79.19, 81.63) --
	( 80.86, 81.76) --
	( 82.53, 81.96) --
	( 84.20, 82.23) --
	( 85.87, 82.55) --
	( 87.54, 82.93) --
	( 89.21, 83.34) --
	( 90.88, 83.82) --
	( 92.55, 84.27) --
	( 94.22, 84.76) --
	( 95.89, 85.26) --
	( 97.56, 85.77) --
	( 99.23, 86.39) --
	(100.90, 86.87) --
	(102.57, 87.48) --
	(104.24, 87.96) --
	(105.91, 88.61) --
	(107.58, 89.25) --
	(109.25, 89.90) --
	(110.92, 90.43) --
	(112.59, 90.94) --
	(114.26, 91.70) --
	(115.93, 92.20) --
	(117.60, 92.84) --
	(119.27, 93.42) --
	(120.94, 94.06) --
	(122.61, 94.66) --
	(124.28, 95.43) --
	(125.95, 95.92) --
	(127.62, 96.68) --
	(129.29, 97.28) --
	(130.96, 97.91) --
	(132.63, 98.35) --
	(134.30, 99.15) --
	(135.97, 99.71) --
	(137.64,100.37) --
	(139.31,101.01) --
	(140.98,101.64) --
	(142.65,102.13) --
	(144.32,102.87) --
	(145.99,103.45) --
	(147.66,104.26) --
	(149.33,104.75) --
	(151.00,105.51) --
	(152.67,106.05) --
	(154.34,106.77) --
	(156.01,107.33) --
	(157.68,108.05) --
	(159.35,108.53) --
	(161.02,109.32) --
	(162.69,109.96) --
	(164.36,110.45) --
	(166.03,111.34) --
	(167.70,111.97) --
	(169.37,112.49) --
	(171.04,113.11) --
	(172.71,113.73) --
	(174.38,114.40) --
	(176.05,115.47) --
	(177.71,115.82) --
	(179.38,116.21) --
	(181.05,117.10) --
	(182.72,117.82) --
	(184.39,118.37) --
	(186.06,119.11) --
	(187.73,119.66) --
	(189.40,120.49) --
	(191.07,120.93) --
	(192.74,121.71) --
	(194.41,122.25) --
	(196.08,123.03) --
	(197.75,123.61) --
	(199.42,124.06) --
	(201.09,124.89) --
	(202.76,125.50) --
	(204.43,126.18) --
	(206.10,126.79) --
	(207.77,127.34) --
	(209.44,128.26) --
	(211.11,128.88) --
	(212.78,129.31) --
	(214.45,130.09) --
	(216.12,130.97) --
	(217.79,131.42) --
	(219.46,132.05) --
	(221.13,132.87);
\end{scope}
\begin{scope}
\path[clip] (  0.00,  0.00) rectangle (252.94,216.81);
\definecolor{drawColor}{RGB}{0,0,0}

\path[draw=drawColor,line width= 0.4pt,line join=round,line cap=round] ( 54.14, 61.20) -- (221.13, 61.20);

\path[draw=drawColor,line width= 0.4pt,line join=round,line cap=round] ( 54.14, 61.20) -- ( 54.14, 55.20);

\path[draw=drawColor,line width= 0.4pt,line join=round,line cap=round] ( 87.54, 61.20) -- ( 87.54, 55.20);

\path[draw=drawColor,line width= 0.4pt,line join=round,line cap=round] (120.94, 61.20) -- (120.94, 55.20);

\path[draw=drawColor,line width= 0.4pt,line join=round,line cap=round] (154.34, 61.20) -- (154.34, 55.20);

\path[draw=drawColor,line width= 0.4pt,line join=round,line cap=round] (187.73, 61.20) -- (187.73, 55.20);

\path[draw=drawColor,line width= 0.4pt,line join=round,line cap=round] (221.13, 61.20) -- (221.13, 55.20);

\node[text=drawColor,anchor=base,inner sep=0pt, outer sep=0pt, scale=  1.00] at ( 54.14, 39.60) {0};

\node[text=drawColor,anchor=base,inner sep=0pt, outer sep=0pt, scale=  1.00] at ( 87.54, 39.60) {2};

\node[text=drawColor,anchor=base,inner sep=0pt, outer sep=0pt, scale=  1.00] at (120.94, 39.60) {4};

\node[text=drawColor,anchor=base,inner sep=0pt, outer sep=0pt, scale=  1.00] at (154.34, 39.60) {6};

\node[text=drawColor,anchor=base,inner sep=0pt, outer sep=0pt, scale=  1.00] at (187.73, 39.60) {8};

\node[text=drawColor,anchor=base,inner sep=0pt, outer sep=0pt, scale=  1.00] at (221.13, 39.60) {10};

\path[draw=drawColor,line width= 0.4pt,line join=round,line cap=round] ( 49.20, 65.14) -- ( 49.20,163.67);

\path[draw=drawColor,line width= 0.4pt,line join=round,line cap=round] ( 49.20, 65.14) -- ( 43.20, 65.14);

\path[draw=drawColor,line width= 0.4pt,line join=round,line cap=round] ( 49.20, 81.56) -- ( 43.20, 81.56);

\path[draw=drawColor,line width= 0.4pt,line join=round,line cap=round] ( 49.20, 97.98) -- ( 43.20, 97.98);

\path[draw=drawColor,line width= 0.4pt,line join=round,line cap=round] ( 49.20,114.40) -- ( 43.20,114.40);

\path[draw=drawColor,line width= 0.4pt,line join=round,line cap=round] ( 49.20,130.83) -- ( 43.20,130.83);

\path[draw=drawColor,line width= 0.4pt,line join=round,line cap=round] ( 49.20,147.25) -- ( 43.20,147.25);

\path[draw=drawColor,line width= 0.4pt,line join=round,line cap=round] ( 49.20,163.67) -- ( 43.20,163.67);

\node[text=drawColor,rotate= 90.00,anchor=base,inner sep=0pt, outer sep=0pt, scale=  1.00] at ( 34.80, 65.14) {-1};

\node[text=drawColor,rotate= 90.00,anchor=base,inner sep=0pt, outer sep=0pt, scale=  1.00] at ( 34.80, 81.56) {0};

\node[text=drawColor,rotate= 90.00,anchor=base,inner sep=0pt, outer sep=0pt, scale=  1.00] at ( 34.80, 97.98) {1};

\node[text=drawColor,rotate= 90.00,anchor=base,inner sep=0pt, outer sep=0pt, scale=  1.00] at ( 34.80,114.40) {2};

\node[text=drawColor,rotate= 90.00,anchor=base,inner sep=0pt, outer sep=0pt, scale=  1.00] at ( 34.80,130.83) {3};

\node[text=drawColor,rotate= 90.00,anchor=base,inner sep=0pt, outer sep=0pt, scale=  1.00] at ( 34.80,147.25) {4};

\node[text=drawColor,rotate= 90.00,anchor=base,inner sep=0pt, outer sep=0pt, scale=  1.00] at ( 34.80,163.67) {5};

\path[draw=drawColor,line width= 0.4pt,line join=round,line cap=round] ( 49.20, 61.20) --
	(227.75, 61.20) --
	(227.75,167.61) --
	( 49.20,167.61) --
	cycle;

\begin{scope}
  \path[draw=drawColor,line width=0.4pt,fill=fillColor,fill opacity=0.85]
    (58,162) rectangle (130,138);

  \draw[line width=0.4pt] (62,156) -- (82,156);
  \node[anchor=west] at (86,156) {\footnotesize CACE};

  \draw[line width=0.4pt,dash pattern=on 1pt off 3pt] (62,149) -- (82,149);
  \node[anchor=west] at (86,149) {\footnotesize N-CACE};

  \draw[line width=0.4pt,dash pattern=on 4pt off 4pt] (62,142) -- (82,142);
  \node[anchor=west] at (86,142) {\footnotesize P-CACE};
\end{scope}

\node[text=drawColor,anchor=base,inner sep=0pt,outer sep=0pt]
  at (137.6, 25) {$W$};

\end{scope}
\end{tikzpicture}
\vspace{-0cm}
\caption{{\bf (Decomposition by P-CACE and N-CACE.)} The solid line is CACE, the dotted line is N-CACE, and the dashed line is P-CACE varying from $W=0$ to $W=10$.
The x-axis represents the value of $W$, and the y-axis represents the values of each measure. }
\label{fig:d1}
\end{figure}

%% file: A2_LL.tex
In this section, we explain the estimation methods.



\subsection{Local Linear Estimator of CDF}
\label{appB1}

Before explaining the estimation methods, we review the local linear (LL) estimator for CDF in \citep{Hall1999}.  
The estimator of $\mathbb{P}(Y< y|X=x,W=w)$ is the solution of $c_0$ by minimizing
\begin{equation}
\sum_{i=1}^N \left\{\mathbb{I}(Y_i < y)-c_0 -c_1(X_i-x)-c_2(W_i-w)\right\}^2K_h(X_i-x)K_h(W_i-w),
\end{equation}
where 
$K_h$ is a kernel function with bandwidth $h$ and $d$ is the dimension of $W$. 
The LL estimator is more resistant to data sparseness than other estimators in \citep{Hall1999}.
Denoting $\alpha(y;x,w)=\mathbb{P}(Y<y|X=x,W=w)$,
we show the consistency conditions given by \citep{Hall1999}.

\paragraph{Condition C1.}
For fixed $y$, $x$, and $w$, $\mathfrak{p}(x,w)>0$, $0<\alpha(y;x,w)<1$, $\mathfrak{p}(x,w)$ is continuous at $(x,w)$, and $\alpha(y;\cdot,\cdot)$ has 
a $2[(d+2)/2]$ continuous derivative in a neighbourhood of $(x,w)$, where $[t]$ denotes the integer part of $t$.

\paragraph{Condition C2.}
The kernel $K$ is a symmetric, compactly supported probability density satisfying $|K(x_1,w_1)-K(x_2,w_2)|\leq C \|(x_1,w_1)-(x_2,w_2)\|$ for $x_1, w_1, x_2, w_2$.

\paragraph{Condition C3.}
The process $\{(X_i,W_i,Y_i)\}$ is absolutely regular; that is,
\begin{equation}
\gamma(j)=\sup_{i\geq 1} \mathbb{E}\Big[\sup_{A \in {\cal F}_{i+1}^{\infty}}\big\|\mathbb{P}(X,W|{\cal F}_1^i)-\mathbb{P}(X,W)\big\| \Big] \rightarrow 0 \text{ as } j\rightarrow \infty,
\end{equation}
where ${\cal F}_i^j$ denotes the $\sigma$ field generated by $\{(X_k,W_k,Y_k):i\leq k\leq j\}$.
Furthermore, $\sum_{j\geq 1}j^2 \gamma(j)^{\delta/(1+\delta)}<\infty$ for some $\delta \in [0,1)$. (Define $a^b=0$ when $a=b=0$.)

\paragraph{Condition C4.}
As $N\rightarrow \infty$, we have $h \rightarrow 0$ and $\lim \inf_{N \rightarrow \infty} Nh^{2(d+1)}>0$.

Under Conditions C1 $\sim$ C4, the error of the local linear estimator $\hat{\alpha}(y;x,w)-\alpha(y;x,w)$ follows $O_p(N^{-1/2}h^{-(d+1)/2}+h^2)$ \citep{Hall1999}.

\subsection{Estimation Methods}
\label{appB2}

Next, we explain the estimators for P-CACE, N-CACE, P-CPICE, and N-CPICE based on the estimator for conditional CDFs in \citep{Hall1999} and Monte Carlo integration \citep{Press2007}.
We assume that we are given an i.i.d. dataset ${\cal D}=\{X_i,Y_i,W_i\}_{i=1}^{N}$ sampled from the distribution $\mathbb{P}(X,Y,W)$ induced by the SCM ${\cal M}$.
We require the following assumption for Monte Carlo integration.
\begin{assumption}[Boundedness of $\Omega_Y$]
\label{BOUD}
The domain $\Omega_Y$ is bounded by $[a,b]$ such that $(-\infty<a<b<\infty)$.
\end{assumption}


\paragraph{Estimation of P-CACE, N-CACE, P-CPICE, and N-CPICE under Identification Assumption.}
We denote the estimator for conditional CDFs using the dataset ${\cal D}$ by $\hat{\alpha}(y;x,w)\defeq \hat{\mathbb{P}}(Y<y|X=x,W=w)$ for any $x \in \Omega_X$, $y \in \Omega_Y$, and $w \in \Omega_W$.
{In this paper, we use the local linear (LL) estimator of CDF in \citep{Hall1999} shown in the previous subsection.
}
We then estimate P-CACE, N-CACE, P-CPICE, and N-CPICE by Monte Carlo integration based on the estimated $\hat{\alpha}$.
We generate $\{y_k\}_{k=1}^{N_2}$ by i.i.d. sampling from a uniform distribution $U[a,b]$ for Monte Carlo integration. 
Then, the estimators of P-CACE and N-CACE are 
\begin{align}
&\hat{\text{\normalfont P-CACE}}(w)=\frac{b-a}{N_2}\sum_{k=1}^{N_2}\max\{\hat{\alpha}(y_k;0,w)-\hat{\alpha}(y_k;1,w),0\},\\
&\hat{\text{\normalfont N-CACE}}(w)=\frac{b-a}{N_2}\sum_{k=1}^{N_2}\max\{\hat{\alpha}(y_k;1,w)-\hat{\alpha}(y_k;0,w),0\}.
\end{align}
Given $w \in \Omega_W$, 
let $\{x_j^0,x_j^1\}_{j=1}^{N_1}$ be i.i.d. samples  from the  conditional joint distribution $\kappa(\cdot\mid w)$.  
The estimators of P-CPICE and N-CPICE are
\begin{align}
&\hat{\text{\normalfont P-CPICE}}(w,\kappa)=\frac{b-a}{N_1N_2}\sum_{j=1}^{N_1}\sum_{k=1}^{N_2}\max\Big\{\hat{\alpha}(y_k;x^0_j,w)-\hat{\alpha}(y_k;x^1_j,w),0\Big\},\\
&\hat{\text{\normalfont N-CPICE}}(w,\kappa)=\frac{b-a}{N_1N_2}\sum_{j=1}^{N_1}\sum_{k=1}^{N_2}\max\Big\{\hat{\alpha}(y_k;x^1_j,w)-\hat{\alpha}(y_k;x^0_j,w),0\Big\}.
\end{align}
P-CACE and N-CACE estimators require ${\cal O}(k(N) \times N_2)$  computation, and P-CPICE and N-CPICE estimators require ${\cal O}(k(N) \times N_1 \times N_2)$ computation, where $k(N)$ is the number of computations to calculate $\hat{\alpha}$.
The lower bounds of our measures without monotonicity are obtained using the same estimators. We clip $\hat{\alpha}$ to $[0,1]$, sort it in $y$ so that it is a valid CDF, and use the same $\hat{\alpha}$ and draws $\{y_k\}$ for every measure and bound.

 
As a side note, CACE and CPICE can be estimated as follows.
\begin{equation}
\begin{aligned}
&\hat{\text{\normalfont CACE}}(w)=\frac{b-a}{N_2}\sum_{k=1}^{N_2}\{\hat{\alpha}(y_k;0,w)-\hat{\alpha}(y_k;1,w)\},
\end{aligned}
\end{equation}
\begin{equation}
\begin{aligned}
&\hat{\text{\normalfont CPICE}}(w,\kappa)=\frac{b-a}{N_1N_2}\sum_{j=1}^{N_1}\sum_{k=1}^{N_2}\{\hat{\alpha}(y_k;x^0_j,w)-\hat{\alpha}(y_k;x^1_j,w)\}.
\end{aligned}
\end{equation}

{If $\hat{\alpha}$ is uniformly consistent over the relevant integration domain as $N \rightarrow \infty$, 
then the estimators of P-CACE, N-CACE, P-CPICE, and N-CPICE are all consistent estimators under  $N_1, N_2 \rightarrow \infty$ by the following theorems.}
\begin{theorem}
\label{Theo3}
Fix $w\in\Omega_W$ and assume that
\begin{equation}
\sup_{\substack{y\in\Omega_Y\\x\in\{0,1\}}}
\left|
\hat{\alpha}(y;x,w)-\alpha(y;x,w)
\right|
=
O_p\!\left(g(N)\right).
\end{equation}
Suppose that the Monte Carlo samples are independent of the data.
Under SCM ${\cal M}$ and Assumptions
\ref{ASEXO2}, \ref{MONO2}, \ref{exi1}, and \ref{BOUD},
\begin{equation}
\hat{\text{\normalfont P-CACE}}(w)
-
\text{\normalfont P-CACE}(w)
=
O_p\!\left(g(N)+N_2^{-1/2}\right),
\end{equation}
and
\begin{equation}
\hat{\text{\normalfont N-CACE}}(w)
-
\text{\normalfont N-CACE}(w)
=
O_p\!\left(g(N)+N_2^{-1/2}\right).
\end{equation}
\end{theorem}

\begin{proof}
Theorem \ref{Theo3} is a special case of Theorem \ref{Theo4}.
\end{proof}

\begin{theorem}
\label{Theo4}
Fix $w\in\Omega_W$ and define
\begin{equation}
\mathcal X_w
=
\operatorname{proj}_1\!\bigl(\operatorname{supp}
\kappa(\cdot,\cdot\mid w)\bigr)
\cup
\operatorname{proj}_2\!\bigl(\operatorname{supp}
\kappa(\cdot,\cdot\mid w)\bigr).
\end{equation}
Assume that
\begin{equation}
\sup_{\substack{y\in\Omega_Y\\x\in\mathcal X_w}}
\left|
\hat{\alpha}(y;x,w)-\alpha(y;x,w)
\right|
=
O_p\!\left(g(N)\right).
\end{equation}
Suppose that the Monte Carlo samples
$\{(x_{0j},x_{1j})\}_{j=1}^{N_1}$ and
$\{y_k\}_{k=1}^{N_2}$ are independent of the data and of each
other, and are generated according to the distributions specified
in the estimators. Under SCM ${\cal M}$ and Assumptions
\ref{ASEXO2}, \ref{MONO2}, \ref{exi3}, and \ref{BOUD},
\begin{equation}
\hat{\text{\normalfont P-CPICE}}(w,\kappa)
-
\text{\normalfont P-CPICE}(w,\kappa)
=
O_p\!\left(
g(N)+N_1^{-1/2}+N_2^{-1/2}
\right),
\end{equation}
and
\begin{equation}
\hat{\text{\normalfont N-CPICE}}(w,\kappa)
-
\text{\normalfont N-CPICE}(w,\kappa)
=
O_p\!\left(
g(N)+N_1^{-1/2}+N_2^{-1/2}
\right).
\end{equation}
\end{theorem}

\begin{proof}
We decompose the error into two parts (A) and (B).
For any $w \in \Omega_W$, we have
\begin{align}
&\text{\normalfont P-CPICE}(w,\kappa)-\hat{\text{\normalfont P-CPICE}}(w,\kappa)\\
&=\frac{b-a}{N_1N_2}\sum_{j=1}^{N_1}\sum_{k=1}^{N_2}\max\Big\{\hat{\mathbb{P}}(Y<y|X=x_0,W=w)-\hat{\mathbb{P}}(Y<y|X=x_1,W=w),0\Big\}\\
&-\int_{\Omega_X}\int_{\Omega_X}\int_{\Omega_Y}\max\Big\{\mathbb{P}(Y<y|X=x_0,W=w)-\mathbb{P}(Y<y|X=x_1,W=w),0\Big\}\\
&\hspace{8cm}dy\kappa(dx_0,dx_1\mid w)\\
&=\frac{b-a}{N_1N_2}\sum_{j=1}^{N_1}\sum_{k=1}^{N_2}\max\Big\{\hat{\mathbb{P}}(Y<y|X=x_0,W=w)-\hat{\mathbb{P}}(Y<y|X=x_1,W=w),0\Big\}\\
&-\int_{\Omega_X}\int_{\Omega_X}\int_{\Omega_Y}\max\Big\{\hat{\mathbb{P}}(Y<y|X=x_0,W=w)-\hat{\mathbb{P}}(Y<y|X=x_1,W=w),0\Big\}\\
&\hspace{8cm}dy\kappa(dx_0,dx_1\mid w)\cdots\text{(A)}\\
&+\int_{\Omega_X}\int_{\Omega_X}\int_{\Omega_Y}\max\Big\{\hat{\mathbb{P}}(Y<y|X=x_0,W=w)-\hat{\mathbb{P}}(Y<y|X=x_1,W=w),0\Big\}\\
&\hspace{8cm}dy\kappa(dx_0,dx_1\mid w)\\
&-\int_{\Omega_X}\int_{\Omega_X}\int_{\Omega_Y}\max\Big\{\mathbb{P}(Y<y|X=x_0,W=w)-\mathbb{P}(Y<y|X=x_1,W=w),0\Big\}\\
&\hspace{8cm}dy\kappa(dx_0,dx_1\mid w)\cdots\text{(B)}
\end{align}

{\bf (A) Error of Monte Carlo integrations.}
We have 
\begin{align}
&\frac{b-a}{N_1N_2}\sum_{j=1}^{N_1}\sum_{k=1}^{N_2}\max\Big\{\hat{\mathbb{P}}(Y<y|X=x_0,W=w)-\hat{\mathbb{P}}(Y<y|X=x_1,W=w),0\Big\}\\
&-\int_{\Omega_X}\int_{\Omega_X}\int_{\Omega_Y}\max\Big\{\hat{\mathbb{P}}(Y<y|X=x_0,W=w)-\hat{\mathbb{P}}(Y<y|X=x_1,W=w),0\Big\}\\
&\hspace{8cm}dy\kappa(dx_0,dx_1\mid w)\\
&=O_p\left(\frac{1}{\sqrt{N_1}}\right)+O_p\left(\frac{1}{\sqrt{N_2}}\right).
\end{align}

{\bf (B) Error of $\alpha$.}
We have 
\begin{align}
&\int_{\Omega_X}\int_{\Omega_X}\int_{\Omega_Y}\max\Big\{\hat{\mathbb{P}}(Y<y|X=x_0,W=w)-\hat{\mathbb{P}}(Y<y|X=x_1,W=w),0\Big\}\\
&\hspace{8cm}dy\kappa(dx_0,dx_1\mid w)\\
&-\int_{\Omega_X}\int_{\Omega_X}\int_{\Omega_Y}\max\Big\{\mathbb{P}(Y<y|X=x_0,W=w)-\mathbb{P}(Y<y|X=x_1,W=w),0\Big\}\\
&\hspace{8cm}dy\kappa(dx_0,dx_1\mid w)\\
&=\int_{\Omega_X}\int_{\Omega_X}\int_{\Omega_Y}\Bigg\{\max\Big\{\hat{\mathbb{P}}(Y<y|X=x_0,W=w)-\hat{\mathbb{P}}(Y<y|X=x_1,W=w),0\Big\}\\
&-\max\Big\{\mathbb{P}(Y<y|X=x_0,W=w)-\mathbb{P}(Y<y|X=x_1,W=w),0\Big\}\Bigg\}\\
&\hspace{9cm}dy\kappa(dx_0,dx_1\mid w).
\end{align}
Let $c=\mathbb{P}(Y<y|X=x_0,W=w)$, $d=\mathbb{P}(Y<y|X=x_1,W=w)$, $\hat{c}=\hat{\mathbb{P}}(Y<y|X=x_0,W=w)$ and $\hat{d}=\hat{\mathbb{P}}(Y<y|X=x_1,W=w)$ for each $y, x_0, w$.
The function $f(c,d)=\max\{c-d,0\}$, which appears inside Term (B), is Lipschitz continuous in both arguments $c$ and $d$ with constant 1.
Then, for any real numbers $c$, $d$, $\hat{c}$, and $\hat{d}$, we have $|\max\{\hat{c}-\hat{d},0\}-\max\{c-d,0\}|\leq |\hat{c}-c|+|\hat{d}-d|$.
Then, we obtain $|\max\{\hat{c}-\hat{d},0\}-\max\{c-d,0\}|=O_p(g(N))$.
By the Lipschitz property and the uniform convergence assumption,
$|B|
\leq
2(b-a)
\sup_{\substack{y\in\Omega_Y\\x\in\mathcal X_w}}
|\hat{\alpha}(y;x,w)-\alpha(y;x,w)|
=
O_p(g(N))$.
Consequently, Term (B) follows $O_p(g(N))$ by integrating $x_0, x_1, y$.

From the sum of (A) and (B), the whole error is bounded by
\begin{align}
&\int_{\Omega_X}\int_{\Omega_X}\int_{\Omega_Y}\max\Big\{\hat{\mathbb{P}}(Y<y|X=x_0,W=w)-\hat{\mathbb{P}}(Y<y|X=x_1,W=w),0\Big\}\\
&\hspace{8cm}dy\kappa(dx_0,dx_1\mid w)\\
&-\int_{\Omega_X}\int_{\Omega_X}\int_{\Omega_Y}\max\Big\{\mathbb{P}(Y<y|X=x_0,W=w)-\mathbb{P}(Y<y|X=x_1,W=w),0\Big\}\\
&\hspace{8cm}dy\kappa(dx_0,dx_1\mid w)\\
&=O_p(g(N))+O_p\left(\frac{1}{\sqrt{N_1}}\right)+O_p\left(\frac{1}{\sqrt{N_2}}\right).
\end{align}
Similarly, $\hat{\text{\normalfont N-CPICE}}(w,\kappa)-\text{\normalfont N-CPICE}(w,\kappa)$ follows
$O_p(g(N))+O_p\left(\frac{1}{\sqrt{N_1}}\right)+O_p\left(\frac{1}{\sqrt{N_2}}\right)$.

\end{proof}


Corollaries~\ref{cor1} and~\ref{cor2} state the resulting rates for any conditional CDF estimator satisfying the required uniform convergence condition. Establishing this condition for a particular estimator requires additional estimator-specific assumptions.

\begin{corollary}\label{cor1}
Under the assumptions of Theorem~\ref{Theo3}, suppose that
\begin{equation}
\sup_{\substack{y\in\Omega_Y\\x\in\{0,1\}}}
|\hat{\alpha}(y;x,w)-\alpha(y;x,w)|
=
O_p(r_N).
\end{equation}
Then the P-CACE and N-CACE estimation errors are
\begin{equation}
O_p\!\left(r_N+N_2^{-1/2}\right).
\end{equation}
\end{corollary}

\begin{corollary}\label{cor2}
Under the assumptions of Theorem~\ref{Theo4}, suppose that
\begin{equation}
\sup_{\substack{y\in\Omega_Y\\x\in\mathcal X_w}}
|\hat{\alpha}(y;x,w)-\alpha(y;x,w)|
=
O_p(r_N).
\end{equation}
Then the P-CPICE and N-CPICE estimation errors are
\begin{equation}
O_p\!\left(r_N+N_1^{-1/2}+N_2^{-1/2}\right).
\end{equation}
\end{corollary}



{
\paragraph{Estimators for Upper Bounds.}
We provide estimators for the upper bounds of our measures.
We generate $\{y_k\}_{k=1}^{N_2}$ by i.i.d. sampling from a uniform distribution $U[a,b]$ for Monte Carlo integration. 
Then, the estimators of the upper bounds of P-CACE and N-CACE are 
\begin{align}
&\hat{\text{\normalfont P-CACE}}^U(w)=\frac{b-a}{N_2}\sum_{k=1}^{N_2}\min\{1-\hat{\alpha}(y_k;1,w),\hat{\alpha}(y_k;0,w)\},\\
&\hat{\text{\normalfont N-CACE}}^U(w)=\frac{b-a}{N_2}\sum_{k=1}^{N_2}\min\{1-\hat{\alpha}(y_k;0,w),\hat{\alpha}(y_k;1,w)\}.
\end{align}
Given $w \in \Omega_W$, 
let $\{x_j^0,x_j^1\}_{j=1}^{N_1}$ be i.i.d. samples from the conditional joint distribution $\kappa(\cdot\mid w)$.  
The estimators of the upper bounds of P-CPICE and N-CPICE are
\begin{align}
&\hat{\text{\normalfont P-CPICE}}^U(w,\kappa)=\frac{b-a}{N_1N_2}\sum_{j=1}^{N_1}\sum_{k=1}^{N_2}\min\Big\{1-\hat{\alpha}(y_k;x^1_j,w),\hat{\alpha}(y_k;x^0_j,w)\Big\},\\
&\hat{\text{\normalfont N-CPICE}}^U(w,\kappa)=\frac{b-a}{N_1N_2}\sum_{j=1}^{N_1}\sum_{k=1}^{N_2}\min\Big\{1-\hat{\alpha}(y_k;x^0_j,w),\hat{\alpha}(y_k;x^1_j,w)\Big\}.
\end{align}
We denote the upper bounds of P-CACE, N-CACE, P-CPICE, and N-CPICE as $\text{\normalfont P-CACE}^U(w)$, $\text{\normalfont N-CACE}^U(w)$, $\text{\normalfont P-CPICE}^U(w)$, and $\text{\normalfont N-CPICE}^U(w)$.
We have the following theorems:
\begin{theorem}
Assume that, over the domain of integration,
\begin{equation}
\sup_{\substack{y\in\Omega_Y\\x\in\{0,1\}}}
\left|
\hat{\alpha}(y;x,w)-\alpha(y;x,w)
\right|
=
O_p\!\left(g(N)\right).
\end{equation}
Suppose that the Monte Carlo samples are independent of the data.
Under SCM ${\cal M}$, Assumptions
\ref{ASEXO2}, \ref{exi1}, and \ref{BOUD}, for any
$w\in\Omega_W$,
\begin{equation}
\hat{\text{\normalfont P-CACE}}^U(w)
-
\text{\normalfont P-CACE}^U(w)
=
O_p\!\left(g(N)+N_2^{-1/2}\right),
\end{equation}
and
\begin{equation}
\hat{\text{\normalfont N-CACE}}^U(w)
-
\text{\normalfont N-CACE}^U(w)
=
O_p\!\left(g(N)+N_2^{-1/2}\right).
\end{equation}
\end{theorem}
\begin{proof}
This theorem is a special case of the next theorem.
\end{proof}
\begin{theorem}
Assume that, over the domain of integration,\begin{equation}
\sup_{\substack{y\in\Omega_Y\\x\in\mathcal X_w}}
\left|
\hat{\alpha}(y;x,w)-\alpha(y;x,w)
\right|
=
O_p\!\left(g(N)\right).
\end{equation}
Suppose that the Monte Carlo samples are independent of the data and
are generated according to the distributions specified in the
estimators. Under SCM ${\cal M}$, Assumptions
\ref{ASEXO2}, \ref{exi3}, and \ref{BOUD}, for any
$w\in\Omega_W$,
\begin{equation}
\hat{\text{\normalfont P-CPICE}}^U(w,\kappa)
-
\text{\normalfont P-CPICE}^U(w,\kappa)
=
O_p\!\left(
g(N)+N_1^{-1/2}+N_2^{-1/2}
\right),
\end{equation}
and
\begin{equation}
\hat{\text{\normalfont N-CPICE}}^U(w,\kappa)
-
\text{\normalfont N-CPICE}^U(w,\kappa)
=
O_p\!\left(
g(N)+N_1^{-1/2}+N_2^{-1/2}
\right).
\end{equation}
\end{theorem}
\begin{proof}
Define
\begin{equation}
\phi_P(c,d)=\min\{1-d,c\},
\qquad
\phi_N(c,d)=\min\{1-c,d\}.
\end{equation}
Both functions are Lipschitz continuous. In particular, for any
$c,d,\tilde c,\tilde d\in[0,1]$,
\begin{equation}
\left|\phi_P(c,d)-\phi_P(\tilde c,\tilde d)\right|
\leq |c-\tilde c|+|d-\tilde d|,
\end{equation}
and the same inequality holds for $\phi_N$.
Decompose the estimation error into the error due to estimating
$\alpha$ and the Monte Carlo integration error. By the assumed
uniform convergence of $\hat{\alpha}$,
\begin{equation}
\sup_{\substack{y\in\Omega_Y\\x\in\mathcal X_w}}
\left|
\hat{\alpha}(y;x,w)-\alpha(y;x,w)
\right|
=
O_p\!\left(g(N)\right).
\end{equation}
The Lipschitz inequalities therefore imply that the plug-in error
for each upper-bound functional is $O_p(g(N))$.

Conditional on the observed data, the integrands are bounded.
Consequently, the Monte Carlo integration error associated with the
$N_1$ draws from $\kappa(\cdot\mid w)$ is $O_p(N_1^{-1/2})$, while
the error associated with the $N_2$ draws over the outcome domain is
$O_p(N_2^{-1/2})$. Combining the plug-in and Monte Carlo errors gives
\begin{equation}
\hat{\text{\normalfont P-CPICE}}^U(w,\kappa)
-
\text{\normalfont P-CPICE}^U(w,\kappa)
=
O_p\!\left(
g(N)+N_1^{-1/2}+N_2^{-1/2}
\right).
\end{equation}
The same argument applied to $\phi_N$ yields
\begin{equation}
\hat{\text{\normalfont N-CPICE}}^U(w,\kappa)
-
\text{\normalfont N-CPICE}^U(w,\kappa)
=
O_p\!\left(
g(N)+N_1^{-1/2}+N_2^{-1/2}
\right).
\end{equation}
\end{proof}
\noindent Then, all estimators of the upper bounds of our measures are consistent.
}

%% file: 6_Experiment.tex
We perform numerical experiments to illustrate the finite-sample properties of the estimators.

\paragraph{Settings.}
We assume the following SCM:
\begin{gather}
Y:=(0.5X+0.1W+1)(-U_Y+0.5), X:=W+3U_X-1,
\end{gather}
where $W$, $U_X$, and $U_Y$ independently follow the uniform distribution $\text{Unif}(0,1)$.
This model satisfies Assumptions \ref{ASEXO2}-\ref{BOUD}.  
Given $W=0.5$, $X$ has support $(-0.5,2.5)$, which contains every treatment value used below. The domain of $Y$ is bounded in $[-1.3,1.3]$.
The key feature of this setting is the presence of an interaction between the treatment $X$ and $U_Y$.
Using this setting, we demonstrate that the proposed measures can uncover heterogeneity resulting from this interaction, which remains unexplained by CACE.

We simulate 100 times with the sample size $N=100,1000,10000$, respectively, and assess the means and 95$\%$ confidence intervals (95$\%$CI) of the estimators.
{We let $N_1$ be 10 and $N_2$ be 1000.}
We use the Epanechnikov kernel with bandwidth $h_j=2\hat\sigma_j N^{-1/6}$ for each regressor in the LL estimator.
The experiments are performed using an Apple M1 (16GB).

\paragraph{Results.}
We investigate the causal effects heterogeneity under the intervention from $X=0$ to $X=2$ for subjects with $W=0.5$.
We investigate the causal effects heterogeneity under stochastic interventions from $X=X^{\pi_0}\sim \text{Unif}(0,0.1)$ to $X=X^{\pi_1}:=X^{\pi_0}+1.9$ for subjects with $W=0.5$.
Table \ref{tab:a} shows the results of experiments.
The results show that, as the sample size increases, the estimates are close to the ground truths and have relatively narrower 95$\%$ CIs.
{Even if the average causal effects after accounting for the subjects' covariate are zero, our newly introduced measures are able to detect substantial heterogeneity over the population. }

{
\paragraph{Experiment for Violation of Monotonicity Assumption.}
We provide additional numerical experiments when the monotonicity assumption (Assumption \ref{MONO2}) is violated.
We assume the following SCM:
\begin{gather}
Y:=(0.5X+0.1W+1)(-U_Y+0.5)+E, X:=W+3U_X-1,
\end{gather}
where $W$, $U_X$, $U_Y$, and $E$ independently follow the uniform distribution $\text{Unif}(0,1)$.
This setting does not satisfy the monotonicity assumption (Assumption \ref{MONO2}) due to $E$.
Table \ref{tab:app2} shows the results under this setting.
The ground truth is the same as the setting of the experiment in the body of the paper.
When the sample size is large ($N=10000$), the average of the estimated bounds contains the ground truth.
}

\begin{table*}[!tb]
\vspace{-0cm}
\centering
\caption{Results of numerical experiments.
}
\label{tab:a}
\vspace{-0.2cm}
\scalebox{0.8}{
\begin{tabular}{c|cccc}
\hline
Estimators & $N=100$ & $N=1000$ & $N=10000$ &  Ground Truth \\
\hline
\hline
CACE & $-0.001$ ([$-0.300,0.333$]) & $0.000$ ([$-0.118,0.131$]) & $0.002$ ([$-0.064,0.065$]) & $0$ \\
P-CACE & $0.155$ ([$0.027,0.383$]) & $0.130$ ([$0.075,0.203$]) & $0.129$ ([$0.099,0.169$]) & $0.125$ \\
N-CACE & $0.156$ ([$0.018,0.334$]) & $0.130$ ([$0.062,0.226$]) & $0.126$ ([$0.095,0.164$]) & $0.125$ \\
\hline
CPICE & $0.000$ ([$-0.268,0.324$]) & $0.001$ ([$-0.120,0.116$]) & $0.002$ ([$-0.066,0.069$]) & $0$ \\
P-CPICE & $0.148$ ([$0.025,0.367$]) & $0.125$ ([$0.071,0.195$]) & $0.122$ ([$0.093,0.163$]) & $0.119$ \\
N-CPICE & $0.148$ ([$0.015,0.296$]) & $0.124$ ([$0.053,0.217$]) & $0.120$ ([$0.087,0.160$]) & $0.119$ \\
\hline
\end{tabular}
}
\vspace{-0.2cm}
\end{table*}

\begin{table}[H]
\vspace{-0cm}
\centering
\caption{Results of numerical experiments when the monotonicity assumption (Assumption \ref{MONO2}) is violated. ``LB" and ``UB" show the estimates of the lower and upper bounds.}
\label{tab:app2}
\vspace{-0cm}
\scalebox{0.8}{
\begin{tabular}{c|cccc}
\hline
Estimators & $N=100$ & $N=1000$ & $N=10000$ &  Ground Truth \\
\hline
\hline
CACE & $0.014$ ([$-0.417,0.466$]) & $0.005$ ([$-0.148,0.167$]) & $-0.003$ ([$-0.075,0.068$]) & $0$ \\
P-CACE (LB) & $0.156$ ([$0.005,0.478$]) & $0.113$ ([$0.041,0.215$]) & $0.105$ ([$0.067,0.143$]) & $0.125$ \\
P-CACE (UB) & $0.446$ ([$0.219,0.682$]) & $0.445$ ([$0.356,0.527$]) & $0.445$ ([$0.395,0.499$]) & $0.125$ \\
N-CACE (LB) & $0.142$ ([$0.004,0.454$]) & $0.108$ ([$0.039,0.191$]) & $0.108$ ([$0.070,0.157$]) & $0.125$ \\
N-CACE (UB) & $0.432$ ([$0.191,0.718$]) & $0.440$ ([$0.347,0.520$]) & $0.448$ ([$0.397,0.498$]) & $0.125$ \\
\hline
CPICE & $0.017$ ([$-0.391,0.438$]) & $0.004$ ([$-0.143,0.163$]) & $-0.002$ ([$-0.067,0.071$]) & $0$ \\
P-CPICE (LB) & $0.150$ ([$0.007,0.447$]) & $0.107$ ([$0.037,0.210$]) & $0.100$ ([$0.064,0.142$]) & $0.119$ \\
P-CPICE (UB) & $0.446$ ([$0.225,0.657$]) & $0.444$ ([$0.358,0.527$]) & $0.445$ ([$0.399,0.497$]) & $0.119$ \\
N-CPICE (LB) & $0.133$ ([$0.002,0.421$]) & $0.103$ ([$0.035,0.185$]) & $0.103$ ([$0.068,0.147$]) & $0.119$ \\
N-CPICE (UB) & $0.429$ ([$0.198,0.690$]) & $0.440$ ([$0.349,0.519$]) & $0.447$ ([$0.397,0.496$]) & $0.119$ \\
\hline
\end{tabular}
}
\vspace{-0cm}
\end{table}

\paragraph{High-dimensional Covariates.}
We demonstrated that our LL estimators perform well in moderate settings.
In practice, researchers may encounter  challenging scenarios involving a large number of covariates.
Regularizations, such as lasso, have been shown to improve the performance of LL estimators in such contexts \citep{Aswani2011}.
The LL estimators with Lasso regularization for P-CACE, N-CACE, P-CPICE, and N-CPICE do not yield extreme values and remain relatively close to the ground truth, even in high-dimensional settings.

{
\paragraph{Setting.}
We assume the following SCM with $D$-dimensional covariates:
\begin{gather}
Y:=(0.5X+0.1\times \bm{1}^T\bm{W}+1)(-U_Y+0.5),
X:=\frac{2(\bm{1}^T\bm{W}+U_X)}{D-2},
\end{gather}
where $U_X$ and $U_Y$ independently follow the uniform distribution $\text{Unif}(0,1)$, $\bm{1}$ is the $D$-dimensional vector of ones,
and $\bm{W}$ independently follows the $D$-dimensional uniform distribution whose variance-covariance matrix is $\bm{\Sigma}$, where $\Sigma_{i,j}=1$ if $i=j$ and $\Sigma_{i,j}=0.9$ if $i\ne j$.
All components of $X$ and $\bm{W}$ are highly correlated with one another.
We compare the LL estimator with no regularization (NR) to its counterpart with Lasso regularization (Lasso), fixing sample size $N = 100$ and varying dimensions $D \in \{5, 30, 100\}$.
We use the R package ``glmnet" for Lasso regularization (\url{https://cran.r-project.org/web/packages/glmnet/index.html}).
We select the regularization parameters using cross-validation implemented in the respective software packages. In this experiment, we choose the bandwidth $h$ of the kernel function used in the LL estimator from the candidates $\{1,0.1,0.01,0.001\}$, using the criterion shown in Section 2.3 of \citep{Hall1999}.
}

{
\paragraph{Results.}
Table \ref{tab:a2} shows the results of experiments.
The LL estimators with no regularization (NR) yield extremely large estimates for P-CACE, N-CACE, P-CPICE, and N-CPICE, particularly in high-dimensional settings, potentially misleading researchers about the extent of causal effect heterogeneity.
The LL estimators with Lasso regularization (Lasso) for P-CACE, N-CACE, P-CPICE, and N-CPICE do not yield extreme values and remain relatively close to the ground truth, even in high-dimensional settings.
}

\begin{table*}[tb]
\centering
\caption{Results of numerical experiments with $D$-dimensional covariates}
\hspace{-0cm}
\label{tab:a2}
\vspace{-0cm}
\scalebox{0.7}{
\begin{tabular}{c|cccc}
\hline
Estimators & $D=5$ & $D=30$ & $D=100$ &  Ground Truth \\
\hline
\hline
CACE (NR)  & $-0.077$ ([$-1.893,1.709$]) & $-0.347$ ([$-33.182,37.904$]) & $298.101$ ([$-8452.903,10660.080$]) &$0$ \\
P-CACE (NR)  & $0.444$ ([$0.000,1.716$]) &$10.075$ ([$0.000,37.904$]) &  $1746.512$ ([$0.000,12087.600$]) &$0.125$ \\
N-CACE (NR)   & $0.520$ ([$0.002,1.903$]) &$10.422$ ([$0.000,33.205$])  &  $1448.411$ ([$0.000,8738.970$])&$0.125$ \\
\hline
CACE (Lasso)  & $-0.008$ ([$-0.739,0.711$]) & $-0.056$ ([$-1.017,0.891$]) & $-0.464$ ([$-3.599,2.467$]) &$0$ \\
P-CACE (Lasso)  & $0.198$ ([$0.015,0.737$]) &$0.347$ ([$0.022,1.005$]) &  $0.910$ ([$0.062,2.698$]) &$0.125$ \\
N-CACE (Lasso)   & $0.206$ ([$0.011,0.807$]) &$0.403$ ([$0.072,1.092$])  &  $1.373$ ([$0.041,3.744$])&$0.125$ \\
\hline
\hline
CPICE(NR)  & $-0.066$ ([$-1.741,1.887$]) & $-0.803$ ([$-44.435,38.003$]) & $208.299$ ([$-6960.519,6675.152$]) &$0$ \\
P-CPICE (NR) & $0.458$ ([$0.000,1.892$]) &$9.978$ ([$0.000,38.003$]) & $2350.278$ ([$0.000,7279.331$]) &$0.119$ \\
N-CPICE (NR)  & $0.524$ ([$0.000,1.751$]) &$10.782$ ([$0.000,44.524$])  &  $2141.979$ ([$0.000,7276.568$])&$0.119$ \\
\hline
CPICE (Lasso) & $0.019$ ([$-0.588,0.757$]) & $-0.068$ ([$-1.185,0.892$]) & $-0.046$ ([$-2.981,2.676$]) &$0$ \\
P-CPICE (Lasso) & $0.196$ ([$0.001,0.812$]) &$0.319$ ([$0.067,0.899$]) & $0.935$ ([$0.059,2.709$]) &$0.119$ \\
N-CPICE  (Lasso) & $0.176$ ([$0.012,0.623$]) &$0.387$ ([$0.014,1.330$])  &  $0.982$ ([$0.002,3.033$])&$0.119$ \\
\hline
\end{tabular}
}
\end{table*}

%% file: A3_Proof.tex
In this section, we provide the proofs of lemmas, propositions, and theorems in the body of the paper.

\Lemmaone*


\begin{proof}
For any $y \in \Omega_Y$ and $w \in \Omega_W$, we have 
\begin{align}
&\mathbb{P}(Y_0<y|W=w)-\mathbb{P}(Y_1<y|W=w)\\
&=\mathbb{P}(Y_0<y,y \leq Y_1|W=w)+\mathbb{P}(Y_0<y,Y_1<y|W=w)\\
&\hspace{3cm}-\mathbb{P}(Y_0<y,Y_1<y|W=w)-\mathbb{P}(Y_1<y,y\leq Y_0|W=w)\\
&=\mathbb{P}(Y_0<y \leq Y_1|W=w)-\mathbb{P}(Y_1<y\leq Y_0|W=w).
\end{align}
\end{proof}

\Propositionone*


\begin{proof}
From Lemma \ref{LEM1}, we have
\begin{align}
\text{\normalfont CACE}(w)&=\mathbb{E}[Y_1|W=w]-\mathbb{E}[Y_0|W=w]\\
&=\int_{\Omega_Y}\{\mathbb{P}(Y_0<y|W=w)-\mathbb{P}(Y_1<y|W=w)\}dy\\
&=\int_{\Omega_Y}\{\mathbb{P}(Y_0<y \leq Y_1|W=w)-\mathbb{P}(Y_1<y\leq Y_0|W=w)\}dy
\end{align}
for each $w \in \Omega_W$.
From Assumption \ref{exi1}, we have
\begin{align}
\text{\normalfont CACE}(w)&=\int_{\Omega_Y}\mathbb{P}(Y_0<y \leq Y_1|W=w)dy-\int_{\Omega_Y}\mathbb{P}(Y_1<y\leq Y_0|W=w)dy\\
&=\text{\normalfont P-CACE}(w)-\text{\normalfont N-CACE}(w)
\end{align}
for each $w \in \Omega_W$.
\end{proof}

\begin{lemma}
\label{lem-app}
Under SCM ${\cal M}$ and Assumptions \ref{ASEXO2} and \ref{MONO2}, for any $x_0, x_1 \in \Omega_X$, we have
\begin{align}
\label{eq83d}
&\mathbb{P}(Y_{x_0}<y \leq Y_{x_1}|W=w)=\max\Big\{\mathbb{P}(Y<y|X={x_0},W=w)-\mathbb{P}(Y<y|X={x_1},W=w),0\Big\}\\
\label{eq84d}
&\mathbb{P}(Y_{x_1}<y\leq Y_{x_0}|W=w)=\max\Big\{\mathbb{P}(Y<y|X={x_1},W=w)-\mathbb{P}(Y<y|X={x_0},W=w),0\Big\}.
\end{align}
\end{lemma}
\begin{proof}
From the identification result of the probability of necessity and sufficiency (PNS) in Theorem 4.2 of \citep{Kawakami2024}, we have 
\begin{align}
&\mathbb{P}(Y_{x_0}<y \leq Y_{x_1}|W=w)=\max\Big\{\mathbb{P}(Y<y|X={x_0},W=w)-\mathbb{P}(Y<y|X={x_1},W=w),0\Big\}\\
&\mathbb{P}(Y_{x_1}<y\leq Y_{x_0}|W=w)=\max\Big\{\mathbb{P}(Y<y|X={x_1},W=w)-\mathbb{P}(Y<y|X={x_0},W=w),0\Big\}.
\end{align}
These expressions follow from the first equation of Theorem~4.2 in \citet{Kawakami2024} by substituting:
(i) $\rho({\boldsymbol y},{\boldsymbol x}_0,{\boldsymbol c})=\mathbb{P}(Y<y|X={x_0},W=w)$  and $\rho({\boldsymbol y},{\boldsymbol x}_1,{\boldsymbol c})=\mathbb{P}(Y<y|X={x_1},W=w)$  for Eq.~\eqref{eq83d} and (ii) $\rho({\boldsymbol y},{\boldsymbol x}_0,{\boldsymbol c})=\mathbb{P}(Y<y|X={x_1},W=w)$  and $\rho({\boldsymbol y},{\boldsymbol x}_1,{\boldsymbol c})=\mathbb{P}(Y<y|X={x_0},W=w)$  for Eq.~\eqref{eq84d}.
\end{proof}

\Theoremone*


\begin{proof}
The definitions of P-CACE and N-CACE are
\begin{align}
&\text{\normalfont P-CACE}(w)=\int_{\Omega_Y} \mathbb{P}(Y_0 < y \leq Y_1|W=w)dy, \\
&\text{\normalfont N-CACE}(w)=\int_{\Omega_Y} \mathbb{P}(Y_1 < y \leq Y_0|W=w)dy.
\end{align}
To obtain the integrands of P-CACE and N-CACE ($\mathbb{P}(Y_0 < y \leq Y_1|W=w)$ and $\mathbb{P}(Y_1 < y \leq Y_0|W=w)$),
we apply Lemma \ref{lem-app} (specialized to the binary treatment case $\Omega_X=\{0,1\}$), which yields
\begin{align}
&\mathbb{P}(Y_0<y \leq Y_1|W=w)=\max\Big\{\mathbb{P}(Y<y|X=0,W=w)-\mathbb{P}(Y<y|X=1,W=w),0\Big\},\\
&\mathbb{P}(Y_1<y\leq Y_0|W=w)=\max\Big\{\mathbb{P}(Y<y|X=1,W=w)-\mathbb{P}(Y<y|X=0,W=w),0\Big\}.
\end{align}
Substituting into the definitions, we obtain
\begin{equation}
\begin{aligned}
&\text{\normalfont P-CACE}(w)=\int_{\Omega_Y} \max\Big\{\mathbb{P}(Y<y|X=0,W=w)-\mathbb{P}(Y<y|X=1,W=w),0\Big\}dy,
\end{aligned}
\end{equation}
\begin{equation}
\begin{aligned}
&\text{\normalfont N-CACE}(w)=\int_{\Omega_Y} \max\Big\{\mathbb{P}(Y<y|X=1,W=w)-\mathbb{P}(Y<y|X=0,W=w),0\Big\}dy.
\end{aligned}
\end{equation}
\end{proof}

\Theoremtwo*


\begin{proof}
For each $y \in \Omega_Y$ and $w \in \Omega_W$, both $\mathbb{P}(Y_0< y \leq Y_1|W=w)$ and $\mathbb{P}(Y_1< y \leq Y_0|W=w)$ are bounded by $l^P(y;w)\leq \mathbb{P}(Y_0< y \leq Y_1|W=w)\leq u^P(y;w)$ and $l^N(y;w)\leq \mathbb{P}(Y_1< y \leq Y_0|W=w)\leq u^N(y;w)$ \citep{Gadbury2004}, where
\begin{equation}
\begin{aligned}
&l^P(y;w)=\max\left\{
\begin{array}{c}
    0,\\
    \mathbb{P}(Y<y|X=0,W=w)-\mathbb{P}(Y<y|X=1,W=w)
\end{array}
\right\},
\end{aligned}
\end{equation}
\begin{equation}
\begin{aligned}
&u^P(y;w)=\min\left\{
\begin{array}{c}
1-\mathbb{P}(Y<y|X=1,W=w),\\
\mathbb{P}(Y<y|X=0,W=w)
\end{array}
\right\},
\end{aligned}
\end{equation}
\begin{equation}
\begin{aligned}
&l^N(y;w)=\max\left\{
\begin{array}{c}
    0,\\
    \mathbb{P}(Y<y|X=1,W=w)-\mathbb{P}(Y<y|X=0,W=w)
\end{array}
\right\},
\end{aligned}
\end{equation}
and
\begin{equation}
\begin{aligned}
&u^N(y;w)=\min\left\{
\begin{array}{c}
1-\mathbb{P}(Y<y|X=0,W=w),\\
\mathbb{P}(Y<y|X=1,W=w)
\end{array}
\right\}.
\end{aligned}
\end{equation}
Since $l^P(y;w)\geq 0$ and $l^N(y;w)\geq 0$ for each $y \in \Omega_Y$ and $w \in \Omega_W$, we have Eqs. \eqref{eq11} and \eqref{eq12}.
The lower bounds are attained when $(Y_0,Y_1)$ are
comonotonic (perfectly positively dependent), whereas the
upper bounds are attained when they are countermonotonic
(perfectly negatively dependent). Thus, these bounds are sharp.
\end{proof}


\Lemmatwo*


\begin{proof}
For any $w \in \Omega_W$, we have 
\begin{align}
&\mathbb{P}(Y_{X^{\pi_0}}<y|W=w)-\mathbb{P}(Y_{X^{\pi_1}}<y|W=w)\\
&=\int_{\Omega_X}\int_{\Omega_X}\Big\{\mathbb{P}(Y_{x_0}<y|W=w)-\mathbb{P}(Y_{x_1}<y|W=w)\Big\}\kappa(dx_0,dx_1\mid w)\\
&=\int_{\Omega_X}\int_{\Omega_X}\Big\{\mathbb{P}(Y_{x_0}<y \leq Y_{x_1}|W=w)-\mathbb{P}(Y_{x_1}<y\leq Y_{x_0}|W=w)\Big\}\kappa(dx_0,dx_1\mid w)\\
&=\mathbb{P}(Y_{X^{\pi_0}}<y\leq Y_{X^{\pi_1}}|W=w)-\mathbb{P}(Y_{X^{\pi_1}}<y\leq Y_{X^{\pi_0}}|W=w).
\end{align}
\end{proof}

\Propositiontwo*


\begin{proof}
From Lemma \ref{LEM3}, we have
\begin{align}
\text{\normalfont CPICE}(w,\kappa)&=\mathbb{E}[Y_{X^{\pi_1}}|W=w]-\mathbb{E}[Y_{X^{\pi_0}}|W=w]\\
&=\int_{\Omega_Y}\{\mathbb{P}(Y_{X^{\pi_0}}<y|W=w)-\mathbb{P}(Y_{X^{\pi_1}}<y|W=w)\}dy\\
&=\int_{\Omega_Y}\{\mathbb{P}(Y_{X^{\pi_0}}<y \leq Y_{X^{\pi_1}}|W=w)-\mathbb{P}(Y_{X^{\pi_1}}<y\leq Y_{X^{\pi_0}}|W=w)\}dy
\end{align}
for each $w \in \Omega_W$.
From Assumption \ref{exi3}, we have
\begin{align}
&\text{\normalfont CPICE}(w,\kappa)\\
&=\int_{\Omega_Y}\mathbb{P}(Y_{X^{\pi_0}}<y \leq Y_{X^{\pi_1}}|W=w)dy-\int_{\Omega_Y}\mathbb{P}(Y_{X^{\pi_1}}<y\leq Y_{X^{\pi_0}}|W=w)dy\\
&=\text{\normalfont P-CPICE}(w,\kappa)-\text{\normalfont N-CPICE}(w,\kappa)
\end{align}
for each $w \in \Omega_W$.
\end{proof}

\Theoremfour*


\begin{proof}
The definitions of P-CPICE and N-CPICE are
\begin{align}
&\text{\normalfont P-CPICE}(w,\kappa)=\int_{\Omega_X}\int_{\Omega_X}\int_{\Omega_Y}\mathbb{P}(Y_{x_0} < y \leq Y_{x_1}|W=w)dy\kappa(dx_0,dx_1\mid w),\\
&\text{\normalfont N-CPICE}(w,\kappa)=\int_{\Omega_X}\int_{\Omega_X}\int_{\Omega_Y}\mathbb{P}(Y_{x_1} < y \leq Y_{x_0}|W=w)dy\kappa(dx_0,dx_1\mid w).
\end{align}
To obtain the integrands of P-CPICE and N-CPICE ($\mathbb{P}(Y_{x_0} < y \leq Y_{x_1}|W=w)$ and $\mathbb{P}(Y_{x_1} < y \leq Y_{x_0}|W=w)$),
we apply  Lemma \ref{lem-app} and substitute them using $\max\{\mathbb{P}(Y<y|X={x_0},W=w)-\mathbb{P}(Y<y|X={x_1},W=w),0\}$ and $\max\{\mathbb{P}(Y<y|X={x_1},W=w)-\mathbb{P}(Y<y|X={x_0},W=w),0\}$ in Eqs. \eqref{eq83d} and \eqref{eq84d}.
Then, we obtain 
\begin{equation}
\begin{aligned}
&\text{\normalfont P-CPICE}(w,\kappa)=\int_{\Omega_X}\int_{\Omega_X}\int_{\Omega_Y}\\
& \max\Big\{\mathbb{P}(Y<y|X=x_0,W=w)-\mathbb{P}(Y<y|X=x_1,W=w),0\Big\}dy\kappa(dx_0,dx_1\mid w),
\end{aligned}
\end{equation}
\begin{equation}
\begin{aligned}
&\text{\normalfont N-CPICE}(w,\kappa)=\int_{\Omega_X}\int_{\Omega_X}\\
&\int_{\Omega_Y} \max\Big\{\mathbb{P}(Y<y|X=x_1,W=w)-\mathbb{P}(Y<y|X=x_0,W=w),0\Big\}dy\kappa(dx_0,dx_1\mid w).
\end{aligned}
\end{equation}
\end{proof}

\Theoremfive*


\begin{proof}
For each $x_0,x_1\in\Omega_X$, $y\in\Omega_Y$, and
$w\in\Omega_W$, the Fr\'echet bounds imply
\begin{align}
l^P(y;x_0,x_1,w)
&\leq
\mathbb{P}(Y_{x_0}<y\leq Y_{x_1}\mid W=w)
\leq
u^P(y;x_0,x_1,w),\\
l^N(y;x_0,x_1,w)
&\leq
\mathbb{P}(Y_{x_1}<y\leq Y_{x_0}\mid W=w)
\leq
u^N(y;x_0,x_1,w),
\end{align}
where $l^P,u^P,l^N,$ and $u^N$ denote the respective integrands in Eqs.~\eqref{eq36}, \eqref{eq21}, \eqref{eq37}, and~\eqref{eq22}.
Integrating these pointwise inequalities with respect to $dy\,\kappa(dx_0,dx_1\mid w)$ yields Eqs.~\eqref{eq21} and~\eqref{eq22}.
For each fixed $(x_0,x_1,y,w)$, the lower Fr\'echet bound is
attained by a comonotonic coupling of $(Y_{x_0},Y_{x_1})$, whereas the upper Fr\'echet bound is attained by a countermonotonic coupling. Thus, the integrands
in Eqs.~\eqref{eq21} and~\eqref{eq22} are pointwise sharp.
The lower integrated bounds are also sharp because the comonotonic construction can be imposed simultaneously over all treatment levels.
We do not claim simultaneous sharpness of the upper integrated bounds without additional restrictions on the joint potential-outcome process.
\end{proof}

%% file: icml2025.bib
@book{Pearl09,
 author = {Judea Pearl},
 title = {Causality: Models, Reasoning and Inference},
 publisher = {Cambridge University Press}, 
 year = {2009},
 edition = {2nd}
}

@article{Imbens1994,
 author = {Guido W. Imbens and Joshua D. Angrist},
 journal = {Econometrica},
 number = {2},
 pages = {467--475},
 publisher = {[Wiley, Econometric Society]},
 title = {Identification and Estimation of Local Average Treatment Effects},
 urldate = {2023-01-07},
 volume = {62},
 year = {1994}
}

@article{Wager2018,
	author = {Stefan Wager and Susan Athey},
	journal = {Journal of the American Statistical Association},
	number = {523},
	pages = {1228-1242},
	title = {Estimation and Inference of Heterogeneous Treatment Effects using Random Forests},
	volume = {113},
	year = {2018}}

@article{Athey2019,
	author = {Athey, Susan and Imbens, Guido W.},
	journal = {Annual Review of Economics},
	number = {1},
	pages = {685-725},
	title = {Machine Learning Methods That Economists Should Know About},
	volume = {11},
	year = {2019}}

@article{Kunzel2019,
	author = {S{\"o}ren R. K{\"u}nzel and Jasjeet S. Sekhon and Peter J. Bickel and Bin Yu},
	journal = {Proceedings of the National Academy of Sciences},
	number = {10},
	pages = {4156-4165},
	title = {Metalearners for estimating heterogeneous treatment effects using machine learning},
	volume = {116},
	year = {2019}}

@article{Nie2020,
	author = {Nie, X and Wager, S},
	journal = {Biometrika},
	month = {09},
	number = {2},
	pages = {299-319},
	title = {{Quasi-oracle estimation of heterogeneous treatment effects}},
	volume = {108},
	year = {2020}}

@article{Athey2016,
	author = {Susan Athey and Guido Imbens},
	journal = {Proceedings of the National Academy of Sciences},
	number = {27},
	pages = {7353-7360},
	title = {Recursive partitioning for heterogeneous causal effects},
	volume = {113},
	year = {2016}}

@article{Zhang2022,
  title={Towards R-learner of conditional average treatment effects with a continuous treatment: T-identification, estimation, and inference},
  author={Zhang, Yichi and Kong, Dehan and Yang, Shu},
  journal={arXiv e-prints},
  pages={arXiv--2208},
  year={2022}
}

@article{Robinson1988,
 author = {P. M. Robinson},
 journal = {Econometrica},
 number = {4},
 pages = {931--954},
 publisher = {[Wiley, Econometric Society]},
 title = {Root-N-Consistent Semiparametric Regression},
 urldate = {2023-03-31},
 volume = {56},
 year = {1988}
}

@article{Singh2023,
	author = {Singh, R and Xu, L and Gretton, A},
	journal = {Biometrika},
	month = {July},
	pages = {asad042},
	title = {{Kernel methods for causal functions: dose, heterogeneous and incremental response curves}},
	year = {2023}}

@InProceedings{Li2023,
  title = 	 {Trustworthy Policy Learning under the Counterfactual No-Harm Criterion},
  author =       {Li, Haoxuan and Zheng, Chunyuan and Cao, Yixiao and Geng, Zhi and Liu, Yue and Wu, Peng},
  booktitle = 	 {Proceedings of the 40th International Conference on Machine Learning},
  pages = 	 {20575--20598},
  year = 	 {2023},
  volume = 	 {202},
  series = 	 {Proceedings of Machine Learning Research},
  month = 	 {23--29 Jul},
  publisher =    {PMLR},
}

@article{Kallus2022,
  title={What's the harm? sharp bounds on the fraction negatively affected by treatment},
  author={Kallus, Nathan},
  journal={Advances in Neural Information Processing Systems},
  volume={35},
  pages={15996--16009},
  year={2022}
}

@article{Frangakis2002,
	author = {Frangakis, Constantine E and Rubin, Donald B},
	journal = {Biometrics},
	month = {Mar},
	number = {1},
	pages = {21--29},
	title = {Principal stratification in causal inference.},
	volume = {58},
	year = {2002}}

@article{Ben-Michael2024,
   author={Eli Ben-Michael and Kosuke Imai and Zhichao Jiang},
	journal = {Journal of the American Statistical Association},
	number = {0},
	pages = {1--14},
	title = {Policy Learning with Asymmetric Counterfactual Utilities},
	volume = {0},
	year = {2024}}

@article{Kawakami2024,
      title={Probabilities of Causation for Continuous and Vector Variables}, 
  author =       {Kawakami, Yuta and Kuroki, Manabu and Tian, Jin},
        year      = {2024},
   journal = {Proceedings of the 40th Conference on Uncertainty in Artificial Intelligence (UAI-2024)},
}

@article{Tian2000,
	author = {Tian, Jin and Pearl, Judea},
	journal = {Annals of Mathematics and Artificial Intelligence},
	number = {1},
	pages = {287--313},
	title = {Probabilities of causation: Bounds and identification},
	volume = {28},
	year = {2000}}

@article{Pearl1999,
	author = {Pearl, Judea},
	journal = {Synthese},
	number = {1},
	pages = {93--149},
	title = {Probabilities Of Causation: Three Counterfactual Interpretations And Their Identification},
	volume = {121},
	year = {1999}}

@article{Munoz2012,
  author={Iv{\'a}n D{\'\i}az Mu{\~n}oz and Van Der Laan, Mark},
 journal = {Biometrics},
 number = {2},
 pages = {541--549},
 publisher = {International Biometric Society},
 title = {Population Intervention Causal Effects Based on Stochastic Interventions},
 urldate = {2024-06-30},
 volume = {68},
 year = {2012}
}

@article{Hall1999,
 author = {Peter Hall and Rodney C. L. Wolff and Qiwei Yao},
 journal = {Journal of the American Statistical Association},
 number = {445},
 pages = {154--163},
 publisher = {[American Statistical Association, Taylor & Francis, Ltd.]},
 title = {Methods for Estimating a Conditional Distribution Function},
 urldate = {2024-06-30},
 volume = {94},
 year = {1999}
}

@article{Mauro2020,
	author = {Mauro, Jacqueline A. and Kennedy, Edward H. and Nagin, Daniel},
	journal = {Journal of the Royal Statistical Society Series A: Statistics in Society},
	month = {04},
	number = {4},
	pages = {1523-1551},
	title = {{Instrumental Variable Methods using Dynamic Interventions}},
	volume = {183},
	year = {2020}}

@article{Diaz2020,
	author = {D{\'\i}az, Iv{\'a}n and Hejazi, Nima S.},
	journal = {Journal of the Royal Statistical Society Series B: Statistical Methodology},
	month = {02},
	number = {3},
	pages = {661-683},
	title = {{Causal Mediation Analysis for Stochastic Interventions}},
	volume = {82},
	year = {2020}}

@article{Kawakami2024b,
  title={Identification and Estimation of Conditional Average Partial Causal Effects via Instrumental Variable},
  author={Kawakami, Yuta and Kuroki, Manabu and Tian, Jin},
   journal = {Proceedings of the 40th Conference on Uncertainty in Artificial Intelligence (UAI-2024)},
  year={2024}
}

@article{Yin2018,
 author = {Yunjian Yin and Lan Liu and Zhi Geng},
 journal = {Statistica Sinica},
 number = {1},
 pages = {115--135},
 publisher = {Institute of Statistical Science, Academia Sinica},
 title = {Assessing the treatment effect heterogeneity with a latent variable},
 urldate = {2024-06-30},
 volume = {28},
 year = {2018}
}

@book{Press2007,
	author = {Press, W.H.},
	publisher = {Cambridge University Press},
	series = {Numerical Recipes: The Art of Scientific Computing},
	title = {Numerical Recipes 3rd Edition: The Art of Scientific Computing},
	year = {2007}}

@article{Poulson2012,
	author = {Poulson, Robert S and Gadbury, Gary L and Allison, David B},
	journal = {Am Stat},
	number = {1},
	pages = {16--24},
	title = {Treatment Heterogeneity and Individual Qualitative Interaction.},
	volume = {66},
	year = {2012}}

@article{Gadbury2000,
 author = {Gary L. Gadbury and Hari K. Iyer},
 journal = {Biometrics},
 number = {3},
 pages = {882--885},
 publisher = {[Wiley, International Biometric Society]},
 title = {Unit-Treatment Interaction and Its Practical Consequences},
 urldate = {2024-07-01},
 volume = {56},
 year = {2000}
}

@article{Gadbury2004,
	author = {Gary L. Gadbury and Hari K. Iyer and Jeffrey M. Albert},
	journal = {Journal of Statistical Planning and Inference},
	number = {2},
	pages = {163-174},
	title = {Individual treatment effects in randomized trials with binary outcomes},
	volume = {121},
	year = {2004}}

@article{Shen2013,
	author = {Shen, Changyu and Jeong, Jaesik and Li, Xiaochun and Chen, Peng-Sheng and Buxton, Alfred},
	journal = {Biometrics},
	month = {Sep},
	number = {3},
	pages = {724--731},
	title = {Treatment benefit and treatment harm rate to characterize heterogeneity in treatment effect.},
	volume = {69},
	year = {2013}}

@article{Zhang2013,
	author = {Zhang, Zhiwei and Wang, Chenguang and Nie, Lei and Soon, Guoxing},
	journal = {J R Stat Soc Ser C Appl Stat},
	month = {Nov},
	number = {5},
	pages = {687--704},
	title = {Assessing the heterogeneity of treatment effects via potential outcomes of individual patients.},
	volume = {62},
	year = {2013}}

@article{Gadbury2001,
	author = {Gadbury, G L and Iyer, H K and Allison, D B},
	journal = {J Biopharm Stat},
	month = {Nov},
	number = {4},
	pages = {313--333},
	title = {Evaluating subject-treatment interaction when comparing two treatments.},
	volume = {11},
	year = {2001}}

@misc{Wu2024,
      title={Quantifying Individual Risk for Binary Outcome: Bounds and Inference}, 
      author={Peng Wu and Peng Ding and Zhi Geng and Yue Liu},
      year={2024},
      eprint={2402.10537},
      archivePrefix={arXiv},
      primaryClass={stat.ME},
      url={https://arxiv.org/abs/2402.10537}, 
}

@article{Xie2013,
	author = {Yu Xie},
	journal = {Proceedings of the National Academy of Sciences},
	number = {16},
	pages = {6262-6268},
	title = {Population heterogeneity and causal inference},
	volume = {110},
	year = {2013}}

@article{Chernozhukov2018,
	author = {Chernozhukov, Victor and Chetverikov, Denis and Demirer, Mert and Duflo, Esther and Hansen, Christian and Newey, Whitney and Robins, James},
	journal = {The Econometrics Journal},
	month = {01},
	number = {1},
	pages = {C1-C68},
	title = {{Double/debiased machine learning for treatment and structural parameters}},
	volume = {21},
	year = {2018}}

@article{Jacob2021,
  title={Cate meets ml: Conditional average treatment effect and machine learning},
  author={Jacob, Daniel},
  journal={Digital Finance},
  volume={3},
  number={2},
  pages={99--148},
  year={2021},
  publisher={Springer}
}

@inproceedings{Shalit2017,
  title={Estimating individual treatment effect: generalization bounds and algorithms},
  author={Shalit, Uri and Johansson, Fredrik D and Sontag, David},
  booktitle={International conference on machine learning},
  pages={3076--3085},
  year={2017},
  organization={PMLR}
}

@book{Lantz2019,
  title={Machine learning with R: expert techniques for predictive modeling},
  author={Lantz, Brett},
  year={2019},
  publisher={Packt publishing ltd}
}

@article{Efron1979,
	author = {B. Efron},
	journal = {The Annals of Statistics},
	number = {1},
	pages = {1 -- 26},
	title = {{Bootstrap Methods: Another Look at the Jackknife}},
	volume = {7},
	year = {1979}}

@article{Powers2018,
	author = {Powers, Scott and Qian, Junyang and Jung, Kenneth and Schuler, Alejandro and Shah, Nigam H and Hastie, Trevor and Tibshirani, Robert},
	journal = {Stat Med},
	month = {May},
	number = {11},
	pages = {1767--1787},
	title = {Some methods for heterogeneous treatment effect estimation in high dimensions.},
	volume = {37},
	year = {2018}}

@article{Kennedy2023,
	author = {Edward H. Kennedy},
	journal = {Electronic Journal of Statistics},
	number = {2},
	pages = {3008 -- 3049},
	title = {{Towards optimal doubly robust estimation of heterogeneous causal effects}},
	volume = {17},
	year = {2023}}

@article{Kato2023,
  title={CATE Lasso: conditional average treatment effect estimation with high-dimensional linear regression},
  author={Kato, Masahiro and Imaizumi, Masaaki},
  journal={arXiv preprint arXiv:2310.16819},
  year={2023}
}

@article{Lo2002,
  title={The true lift model: a novel data mining approach to response modeling in database marketing},
  author={Lo, Victor SY},
  journal={ACM SIGKDD Explorations Newsletter},
  volume={4},
  number={2},
  pages={78--86},
  year={2002},
  publisher={ACM New York, NY, USA}
}

@article{Manski2004,
	author = {Manski, Charles F.},
	journal = {Econometrica},
	number = {4},
	pages = {1221-1246},
	title = {Statistical Treatment Rules for Heterogeneous Populations},
	volume = {72},
	year = {2004}}

@article{Lou2018,
	author = {Lou, Zhilan and Shao, Jun and Yu, Menggang},
	journal = {Biometrics},
	month = {Jun},
	number = {2},
	pages = {506--516},
	title = {Optimal treatment assignment to maximize expected outcome with multiple treatments.},
	volume = {74},
	year = {2018}}

@incollection{Brand2013,
  title={Causal effect heterogeneity},
  author={Brand, Jennie E and Thomas, Juli Simon},
  booktitle={Handbook of causal analysis for social research},
  pages={189--213},
  year={2013},
  publisher={Springer}
}

@article{Albert2005,
	author = {Albert, Jeffrey M. and Gadbury, Gary L. and Mascha, Edward J.},
	journal = {Biometrical Journal},
	number = {5},
	pages = {662-673},
	title = {Assessing Treatment Effect Heterogeneity in Clinical Trials with Blocked Binary Outcomes},
	volume = {47},
	year = {2005}}

@phdthesis{Post2023,
  title={Causal Effect Heterogeneity: Statistical formalization and analysis of the individual causal effect},
  author={Post, R. A. J.},
  year={2023},
  school={Eindhoven University of Technology},
  type={{Phd Thesis 1 (Research TU/e / Graduation TU/e), Mathematics and Computer Science}},
}

@article{Dudik2014,
 author = {Miroslav Dudík and Dumitru Erhan and John Langford and Lihong Li},
 journal = {Statistical Science},
 number = {4},
 pages = {485--511},
 publisher = {Institute of Mathematical Statistics},
 title = {Doubly Robust Policy Evaluation and Optimization},
 urldate = {2024-10-25},
 volume = {29},
 year = {2014}
}

@inproceedings{Oberst2019,
  title={Counterfactual off-policy evaluation with gumbel-max structural causal models},
  author={Oberst, Michael and Sontag, David},
  booktitle={International Conference on Machine Learning},
  pages={4881--4890},
  year={2019},
  organization={PMLR}
}

@inproceedings{Kallus2020,
  title={Statistically efficient off-policy policy gradients},
  author={Kallus, Nathan and Uehara, Masatoshi},
  booktitle={International Conference on Machine Learning},
  pages={5089--5100},
  year={2020},
  organization={PMLR}
}

@inproceedings{Saito2023,
  title={Off-policy evaluation for large action spaces via conjunct effect modeling},
  author={Saito, Yuta and Ren, Qingyang and Joachims, Thorsten},
  booktitle={International Conference on Machine Learning},
  pages={29734--29759},
  year={2023},
  organization={PMLR}
}

@article{Aswani2011,
 author = {Anil Aswani and Peter Bickel and Claire Tomlin},
 journal = {The Annals of Statistics},
 number = {1},
 pages = {48--81},
 publisher = {Institute of Mathematical Statistics},
 title = {Regression on manifolds: Estimation of the exterior derivative},
 urldate = {2025-04-02},
 volume = {39},
 year = {2011}
}

@article{Pearl2017,
  title={Detecting Latent Heterogeneity.},
  author={Pearl, Judea},
  journal={Sociological Methods \& Research},
  volume={46},
  number={3},
  pages={370--389},
  year={2017},
  publisher={ERIC}
}
